\documentclass[12pt]{amsart}

\usepackage{amsmath,amsthm,amsfonts,amssymb}
\usepackage{color,graphicx}

\newtheorem{thm}{Theorem}[section]
\newtheorem{df}[thm]{Definition}

\newtheorem{lem}[thm]{Lemma}

\newtheorem{rem}[thm]{Remark}

\newenvironment{pf1}{\noindent {\it Proof }}{\hspace*{\fill} $\Box$ \\}

\begin{document}

\title[Chern--Simons--Schr\"{o}dinger theory on a lattice]{Chern--Simons--Schr\"{o}dinger theory \\on a one-dimensional lattice}

\author[H. Huh]{Hyungjin Huh}
\address[H. Huh]{Department of Mathematics, Chung-Ang University, Seoul, Republic of Korea, 06974}
\email{huh@cau.ac.kr}

\author[S. Hussain]{Swaleh Hussain}
\address[S.Hussain]{Department of Mathematics and Statistics, McMaster University, Hamilton, Ontario, Canada, L8S 4K1}
\email{hussainswaleh@gmail.com}

\author[D.E. Pelinovsky]{Dmitry E. Pelinovsky$^{*}$}\footnote{Corresponding author.}
\address[D.E. Pelinovsky]{Department of Mathematics and Statistics, McMaster University, Hamilton, Ontario, Canada, L8S 4K1}
\email{dmpeli@math.mcmaster.ca}

\begin{abstract}
We propose a gauge-invariant system of the Chern--Simons--Schr\"{o}dinger type on a one-dimensional lattice.
By using the spatial gauge condition, we prove local and global well-posedness of the initial-value problem
in the space of square summable sequences for the scalar field. We also study the existence region of
the stationary bound states, which depends on the lattice spacing and the nonlinearity power. A major difficulty in
the existence problem is related to the lack of variational formulation of the stationary equations.
Our approach is based on the implicit function theorem in the anti-continuum limit and the solvability constraint
in the continuum limit.
\end{abstract}



\maketitle
\date{}

\section{Introduction}
\setcounter{equation}{0}

Gauge theories are important in quantum electrodynamics, quantum chromodynamics, and particle physics.
In quantum chromodynamics, perturbative calculations break down frequently in the high energy regime
resulting in the so-called {\em ultraviolet divergence}, or the divergence at small lengths.
Non-perturbative calculations formally involve evaluating an infinite-dimensional path integral
which is computationally intractable. To overcome the divergence problem,
Wilson  developed lattice gauge theory by working on lattice with the smallest length determined by the lattice spacing $h$
\cite{W}. The path integral becomes finite-dimensional on lattice and thus can be easily evaluated. When $h$ goes to zero,
the lattice gauge theory converges to the continuum gauge theory at the formal level. See \cite{Gatt,Ko,Smit} for review.
The lattice gauge theory attracted a lot of attention of physicists and mathematicians
(see \cite{BG, Cha, GR} for recent studies).

Dynamics of matter and gauge fields can be described by several types of models which
include nonlinear wave, Schr\"{o}dinger, Dirac, and Ginzburg--Landau equations with either Maxwell or Chern--Simon gauge.
Our work corresponds to the case of the nonlinear Schr\"{o}dinger equation with the Chern--Simon gauge, which we label as
the CSS system.

In the continuous setting, the initial-value problem of the CSS system was studied in \cite{BBS, LST}
and the stationary bound states of the CSS system were constructed in \cite{BHS, PR}. The main objective of this work is
to propose a gauge-invariant discretization of the CSS system on a one-dimensional grid with the lattice spacing $h>0$
and to study both the initial-value problem and the existence of stationary bound states in the discrete CSS system.

The concepts of gauge invariance and preservation of the gauge constraints are crucial elements in
the study of gauged nonlinear evolution equations. For instance,  the initial-value problem of the
nonlinear Schr\"{o}dinger equation with the Maxwell gauge was studied in \cite{GNS} by considering equations
with a dissipation term, which was added to preserve the constraint equation. As the dissipation term vanishes,
conservation of energy and charge was used to obtain compactness.

The numerical studies of the gauged evolution equations are mostly confined to the conventional finite difference
and finite element methods \cite{ LZ, MC, RS, WS}. In the last few decades, structure-preserving discretization  \cite{Ch3} has emerged
as an important tool of the numerical computations. The gauge invariant difference
approximation of the Maxwell gauged equations was studied in \cite{Ch1, Ch2, Du2, Du3}.
In particular, it was shown in \cite{Ch2} that the discretized solution of the finite element method
with a gauge constraint converges to a weak solution of the Maxwell-Klein-Gordon equation in two space dimensions
for initial data of finite energy. The essential features of the discretization were the energy conservation
and the constraint preservation, which give control over the curl and divergence of the vector potential.

The discrete CSS system, which we consider here, is also based on the finite-difference method
and the discretization is proposed in such a way that the CSS system remains gauge invariant
with the gauge constraint being preserved in the time evolution. This allows us to simplify the system of equations
to the discrete NLS (nonlinear Schr\"{o}dinger) equation for the scalar field
coupled with the constraints on components of the gauge vector. Local well-posedness of the
initial-value problem of the discrete CSS system follows from this constrained
discrete NLS equation by the standard fixed-point arguments.

We show that the time-evolution of the discrete CSS system preserves the mass defined as the squared $\ell^2$ norm
of the scalar field. However, the total energy is not preserved in the time evolution. Nevertheless, the mass
conservation is sufficient in order to extend the local solutions for all times and
to conclude on the global well-posedness of the initial-value problem of the discrete CSS system.

The lack of energy conservation presents difficulties in constructing the stationary bound states of
the discrete CSS system with a variational approach. As a result, we construct the stationary bound states
by using the implicit function theorem in the anti-continuum limit as $h \to \infty$ at least for sub-quintic
powers of the nonlinearity. For quintic and super-quintic nonlinearities, the stationary solutions do not usually extend
to the anti-continuum limit and terminate at the fold points. We also show that the stationary solutions
do not extend to the continuum limit as $h \to 0$ for any nonlinearity powers and terminate at the fold points.
These analytical results are complemented with the numerical approximations of the stationary bound states
continued with respect to the lattice spacing parameter $h$.

The anti-continuum limit $h \to \infty$ is a popular case of study, for which
the existence of stationary bound states can be proven with analytical tools \cite{MA94}. This limit
corresponds to the weakly coupled lattice and is opposite to the continuum limit $h \to 0$,
for which the lattice formally converges to the continuous system.

There are technical obstacles to explore the analogous questions on a two-dimensional lattice.
The gauge constraints do not allow us to simplify the discrete CSS system to the form of the constrained
discrete NLS equation.

The article is organized as follows. Section \ref{sec-2} presents the main results.
Well-posedness of the initial-value problem is considered in Section \ref{sec-3}.
Analytical results on the existence of stationary bound states are proven in Section \ref{sec-4}.
Numerical approximations of the stationary bound states are collected in Section \ref{sec-5}.
Section \ref{sec-6} concludes the article with a summary.

\section{Main results}
\setcounter{equation}{0}
\label{sec-2}

The continuous CSS system in two space dimensions can be written in the following form:
\begin{align}\label{css-two-dimensions}
\left\{
 \begin{aligned}
& i D_0 \phi + D_1D_1 \phi + D_2 D_2\phi+ \lambda |\phi|^{2p}\,\phi=0, \\
&  \partial_t A_1 - \partial_x A_0 = \mbox{Im} (\bar{\phi} D_2 \phi),\\
& \partial_t A_2 - \partial_y A_0 = -\mbox{Im} (\bar{\phi} D_1 \phi),\\
& \partial_x A_2 - \partial_y A_1 = \frac12 |\phi|^2,
 \end{aligned} \right.
\end{align}
where $t \in \Bbb{R}^1$, $(x,\,y) \in \Bbb{R}^2$, $\phi \in \Bbb{C}$ is the scalar field,
$(A_{0}, \, A_1, \, A_2) \in \Bbb{R}^3$  is the gauge vector,
$D_{0} = \partial_t - i A_0$, $D_1 = \partial_x - i A_1$, and $D_2 = \partial_y - i A_2$  are the covariant derivatives,
$\lambda>0$ is a coupling constant representing the strength of interaction potential,
and $p > 0$ is the nonlinearity power. The CSS system (\ref{css-two-dimensions})
admits a Hamiltonian formulation with conserved mass and total energy \cite{Ham1,Ham2}.

When the scalar field and the gauge vector are independent of $y$, the continuous CSS system \eqref{css-two-dimensions}
can be rewritten in one space dimension as follows:
\begin{align}\label{css}
\left\{
 \begin{aligned}
& i D_0 \phi + D_1D_1 \phi-A_2^2 \phi+ \lambda |\phi|^{2p}\,\phi=0, \\
&  \partial_t A_1 - \partial_x A_0 = -A_2| \phi|^2,\\
& \partial_t A_2 = -\mbox{Im} (\bar{\phi} D_1 \phi),\\
& \partial_x A_2 =\frac12 |\phi|^2,
 \end{aligned} \right.
\end{align}
where $\phi(t,x): \Bbb{R}^1 \times \Bbb{R}^1 \to \Bbb{C}$ and $(A_{0}, \, A_1, \, A_2)(t,x) : \Bbb{R}^1 \times \Bbb{R}^1 \to \Bbb{R}^3$.
The continuous one-dimensional CSS system \eqref{css} admits conservation of mass
\begin{equation}
\label{mass}
M = \int_{\mathbb{R}} |\phi(t,x)|^2 dx
\end{equation}
and conservation of the total energy
\begin{equation}
\label{energy}
E = \int_{\mathbb{R}}\Big(  |D_1 \phi|^2 + A_2^2 |\phi|^2 - \frac{\lambda}{p+1} |\phi|^{2p+2} \Big) (t,x) \,dx.
\end{equation}

We propose to consider the following discrete CSS system:
\begin{align}\label{dcss}
\left\{
 \begin{aligned}
& i D_0 \phi(t, n) +D_1^- D_1^+ \phi(t, n) - A_2^2(t, n) \phi(t, n)+ \lambda |\phi(t, n)|^{2p}\,
\phi(t, n)=0, \\
&  \partial_t A_1(t, n+\frac12) -\nabla_1^{+} A_0 (t, n)= - A_2(t, n) |\phi(t, n)|^2, \\
& \partial_t A_2 (t, n)= -\mbox{Im} (\bar{\phi}(t,n-1) D_1^{+} \phi (t,n-1)), \\
& \nabla_1^{+} A_2 (t, n) =\frac12 |\phi (t, n)|^2,
\end{aligned} \right.
\end{align}
where $n\in \Bbb{Z}$, $\phi(t,\,n)$, $A_0(t, \,n)$, $A_2(t, \,n)$ are defined on lattice sites $n$, and
$A_1(t, \,n+ \frac{1}{2})$ is defined at middle distance between lattice sites $n$ and $n+1$.
Similarly to the continuous CSS system \eqref{css}, $\phi(t, n )$ is the scalar field, whereas
$A_{0}(t, n ), \, A_1(t,n+\frac{1}{2}), \, A_2(t, n)$  are components of the gauge vector.
The discrete covariant derivatives are defined as
\begin{align}
\label{covariant}
\left\{
 \begin{aligned}
D_1^{+}\phi(t,\,n)&=\frac{1}{h}\left[ e^{-ih A_1(t,\,n+\frac12)}  \phi(t,\,n+1)-\phi(t,\,n) \right],\\
D_1^{-}\phi(t,\, n)&=\frac{1}{h}\left[ \phi(t, n)-e^{ih A_1(t,\, n-\frac12)} \phi(t,\, n-1) \right],
\end{aligned}  \right.
\end{align}
whereas the finite difference operators are defined by
\begin{align}
\label{difference}
\left\{
 \begin{aligned}
\nabla_1^{+}f (t,\, n)=\frac{ 1}{h} \Big[     f(t,\, n+1) -  f(t,\,n)    \Big],\\
\nabla_1^{-}f (t,\, n)=\frac{ 1}{h} \Big[     f(t,\,n) -  f(t,\,n-1)    \Big]. \end{aligned}  \right.
\end{align}
In the continuum limit $h\to 0$, if $f(t,n)$, $n \in \Bbb{Z}$ converges to a smooth function $\textsf{f}(t,x)$, $x \in \mathbb{R}$
such that $f(t,n) = \textsf{f}(t,hn)$ for every $n \in \Bbb{Z}$, then the discrete covariant derivatives \eqref{covariant} and
the finite differences \eqref{difference} converge formally at any fixed $x \in \mathbb{R}$:
\begin{align*}
\left\{
 \begin{aligned}
& D_1^{\pm} f(t,n) \to D_1 \textsf{f}(t,x),  \\
& \nabla_1^{\pm} f(t,n) \to \partial_x \textsf{f}(t,x), \end{aligned}  \right.
\end{align*}
where the lattice is centered at fixed $x$. The continuous CSS system \eqref{css} follows
formally from the discrete CSS system \eqref{dcss} as $h \to 0$.

It is natural to look for solutions to the discrete CSS system \eqref{dcss} in the
space of squared summable sequences for the sequence $\{ \phi(n) \}_{n \in \mathbb{Z}}$ denoted simply by $\phi$:
\[
\ell^2(\Bbb{Z}) := \Big\{ \phi \in \Bbb{C}^{\Bbb{Z}} : \quad   \| \phi \|_{\ell^2}^2  := \sum_{n\in \Bbb{Z}} |\phi(n)|^2 < \infty \Big\},
\]
equipped with the inner product
\[
\langle \phi, \psi   \rangle = \sum_{n \in \mathbb{Z}} \phi (n) \overline{\psi} (n).
\]
The $\ell^2$ space is embedded into $\ell^p$ spaces for every $p > 2$ in the sense
of $\| \phi \|_{\ell^p} \leq \| \phi \|_{\ell^2}$. The embedding includes
the limiting case $p = \infty$ for which $\|\phi \|_{\ell^{\infty}} =\sup_{n \in \mathbb{Z}} |\phi(n)|$.

\begin{df}
\label{def-sol}
We say that $(\phi,A_0,A_1,A_2)$ is a local solution to the discrete CSS system \eqref{dcss}
if there exists $T > 0$ such that
\begin{align*}
\phi\in C^1([-T, \,T],\, \ell^2(\mathbb{Z})) \quad \mbox{and} \quad  A_0, \, A_1,\, A_2 \in C^1([-T, \,T],\, \ell^{\infty}(\mathbb{Z}))
\end{align*}
satisfy \eqref{dcss}. If $T > 0$ can be extended to be arbitrary large, then we say that
$(\phi,A_0,A_1,A_2)$ is a global solution to the discrete CSS system \eqref{dcss}.
\end{df}

A local solution to the discrete CSS system \eqref{dcss} in the sense of Definition \ref{def-sol} enjoys conservation of the mass
\begin{equation}
\label{mass-discrete}
M = h \sum_{n \in \mathbb{Z}} |\phi(t,n)|^2
\end{equation}
which generalizes the mass \eqref{mass} of the continuous CSS system \eqref{css}. On the other hand, no conservation
of energy exists in the discrete CSS system \eqref{dcss}, which would generalize the energy \eqref{energy} of the continuous CSS system \eqref{css}.
See Remarks \ref{rem-redundancy} and \ref{rem-mass}.

Because the last equation of the discrete CSS system \eqref{dcss} is redundant in the initial-value problem (Lemma \ref{lem-redundancy}),
the local well-posedness of the initial-value problem can not be established without a gauge condition. However,
the discrete CSS system \eqref{dcss} enjoys the gauge invariance (Lemma \ref{lem-gauge}) and this invariance
can be used to reformulate the discrete CSS system \eqref{dcss} with the gauge condition $A_1(t,n+\frac{1}{2}) = 0$.
The simpler form \eqref{dcssr} of the discrete CSS system consists of the NLS equation for the scalar field $\phi$
constrained by two equations on $A_0$ and $A_2$. The following theorem represents the main result on global well-posedness
of the initial-value problem for the discrete CSS system \eqref{dcss} with the gauge condition $A_1(t,n+\frac{1}{2}) = 0$.

\begin{thm}
\label{lem-WPS}
For every $\Phi \in \ell^2(\mathbb{Z})$ and every $(\alpha,\beta)$,
there exists a unique global solution $(\phi,\,A_0,\,A_1 = 0,\,A_2) $ to the initial-value problem
for the discrete CSS system \eqref{dcss} satisfying
\begin{equation}
\label{solutions-WPS}
\phi\in C^1 (\mathbb{R}, \, \ell^2(\mathbb{Z})) \quad \mbox{ and } \quad A_0, \,A_2 \in C^1 (\mathbb{R}, \, \ell^{\infty}(\mathbb{Z})),
\end{equation}
the initial conditions $\phi(0, n)=\Phi_n$, $A_0(0,n) = \mathcal{A}_n$ and $A_2(0,n) = \mathcal{B}_n$,
the boundary conditions
\begin{equation}
\label{boundary-values}
\lim_{n\to \infty} A_0(t, n)=\alpha \quad \mbox{\rm and} \quad
\lim_{n\to \infty} A_2(t, n)=\beta,
\end{equation}
and the consistency conditions
\begin{equation}
\label{consistency}
(\nabla_1^+ \mathcal{A})_n = \mathcal{B}_n |\Phi_n|^2, \quad (\nabla_1^+ \mathcal{B})_n = \frac{1}{2} |\Phi_n|^2.
\end{equation}
Moreover, the solution \eqref{solutions-WPS} depends continuously on $\Phi \in \ell^2(\mathbb{Z})$.
\end{thm}

The continuous CSS system \eqref{css} with the gauge condition $A_1(t,x) = 0$ can be reduced to
the continuous NLS equation \eqref{conti} for $\varphi(t,x)$ \cite{Huh}.
The continuous NLS equation admits a family of stationary bound states
$\varphi(t,x) = Q(x) e^{i \omega t}$ for every $\lambda > 0$, $p > 0$, and $\omega > 0$,
where $Q(x)$ can be written in the explicit form:
\begin{equation}
\label{soliton}
Q(x) = \sqrt{\omega (p+1) \lambda^{-1}} {\rm sech}^{\frac{1}{p}}(\sqrt{\omega} p x).
\end{equation}
It is natural to ask if the discrete CSS system \eqref{dcss} also admits stationary bound states for $\lambda > 0$ and $p > 0$.
The existence problem for stationary bound states reduces to the system of difference equations \eqref{equ23}.
It is rather surprising that the existence of stationary bound states of the discrete CSS system \eqref{dcss}
depends on the values of parameters $p$ and $h$.

The following theorem represents the main result on
the existence of stationary bound states of the discrete CSS system \eqref{dcss} in the anti-continuum limit $h \to \infty$.
The stationary bound states decay fast as $|n| \to \infty$, therefore, their existence
can be proven in the space of summable sequences denoted by $\ell^1(\Bbb{Z})$ with the norm
$\| \phi \|_{\ell^1} = \sum_{n \in \mathbb{Z}} |\phi(n)|$.

\begin{thm}\label{thm-existence}
For every $\lambda > 0$, $p \in (0,2)$, $\Omega > 0$,
and sufficiently large $h$, there exists a unique family of stationary bound states
in the form
\begin{equation}
\label{bound-states}
\phi(t,n) = e^{i \Omega t} U_n, \quad g(t,n) = G_n, \quad t \in \Bbb{R}, \quad n \in \mathbb{Z}
\end{equation}
with $U \in \ell^1(\Bbb{Z})$ and $G \in \ell^{\infty}(\Bbb{Z})$ such that
\[
U \sim \mathcal{U}(h) \delta_{0}, \quad G \sim \frac{h^2}{4} \mathcal{U}(h)^4 \chi_{0} \quad \mbox{\rm as} \quad h \to \infty,
\]
where $\mathcal{U}(h)$ is a positive root of the nonlinear equation
\begin{equation}
\label{root-scalar}
\lambda \mathcal{U}^{2p} + \frac{h^2}{4} \mathcal{U}^4 - \Omega = 0.
\end{equation}
Here the sign $\sim$ means the asymptotic expansion with the next-order term being smaller as $h \to \infty$ compared to the
leading-order term in the $\ell^{\infty}$ norm. The discrete $\delta_0$ and $\chi_0$ are defined by their components$:$
\begin{equation}
\label{delta}
[\delta_{0}]_n = \left\{ \begin{array}{l} 1, \; n = 0, \\ 0, \; n \neq 0, \end{array} \right. \quad
[\chi_{0} ]_n = \left\{ \begin{array}{l} 1, \; n \leq 0, \\ 0, \; n > 0. \end{array} \right.
\end{equation}
\end{thm}

In order to study the anti-continuum limit $h \to \infty$, we use the implicit function theorem similar
to the study of weakly coupled lattices in the anti-continuum limit \cite{MA94}. In particular,
we reformulate the difference equations \eqref{equ23} as the root-finding problem (Lemma \ref{lem-1}),
study the asymptotic behavior of roots $\mathcal{U}(h)$ in the nonlinear equation \eqref{root-scalar} (Lemma \ref{lem-0}),
and find the unique continuation of the single-site solutions with respect to the small
parameter (Lemma \ref{lem-2}). Compared with the standard anti-continuum limit in \cite{MA94},
the root $\mathcal{U}(h)$ of the nonlinear equation \eqref{root-scalar} depends on $h$
and the Jacobian of the difference equations \eqref{equ23} becomes singular as $h \to \infty$,
therefore, we need to use a renormalization technique in order to prove Theorem \ref{thm-existence}.
Besides the single-site solutions in Theorem \ref{thm-existence}, one can use the same technique
and justify the double-site and generally multi-site solutions in the anti-continuum limit (Remark \ref{rem-two-site}).

Another interesting and surprising result is that the stationary bound states of the discrete CSS system
\eqref{dcss} do not converge to the stationary bound states \eqref{soliton} of the continuous CSS system
\eqref{css} in the continuum limit $h \to 0$. The following theorem gives the corresponding result
which is proven with the use of the solvability constraint on suitable solutions to the difference equations \eqref{equ23}.

\begin{thm}
Let $Q$ be defined by \eqref{soliton}.
For every $\lambda > 0$, $p > 0$, $\Omega > 0$, and sufficiently small $h$, there exist no stationary bound states
in the form \eqref{bound-states} such that $U \in \ell^1(\Bbb{Z})$ and $G \in \ell^{\infty}(\Bbb{Z})$ satisfy
$U_{-n} = U_n$, $n \in \Bbb{Z}$ and
\begin{equation}
\label{error-bound}
 \| U - Q(h \cdot) \|_{\ell^1} + \| G \|_{\ell^{\infty}} \leq C h
\end{equation}
for an $h$-independent positive constant $C$.
\label{thm-nonexistence}
\end{thm}

Because the difference equations for the stationary bound states \eqref{bound-states}
do not allow a variational formulation due to the lack
of energy conservation, we are not able to study the existence problem
in the entire parameter region. However, we show numerically in Section \ref{sec-5}
by using the parameter continuation in $h$ that the single-site bound states of Theorem \ref{thm-existence} for $p \in (0,2)$
do not continue to the limit $h \to 0$ in accordance with Theorem \ref{thm-nonexistence}
because of the fold bifurcation with another family of stationary bound states.
The other family converges to the double-site solution in the anti-continuum limit $h \to \infty$ for $p \in (0,2)$.
Moreover, for $p \geq 2$, we show that the family of single-site bound states do not
continue in both limits $h \to 0$ and $h \to \infty$ because of fold bifurcations in each direction of $h$.

\section{Well-posedness of the discrete CSS system}
\setcounter{equation}{0}
\label{sec-3}

Here we consider well-posedness of the Cauchy problem associated with the discrete CSS system \eqref{dcss}.
In the end, we will prove Theorem \ref{lem-WPS}.

The discrete CSS system \eqref{dcss} consists of four equations for four unknowns, however,
the time evolution is only defined by the first three equations whereas
the last equation is a constraint. The following lemma states that
this constraint is invariant with respect to the time evolution.

\begin{lem}
\label{lem-redundancy}
Assume that the initial data satisfies
\begin{equation}
\label{initial-id}
\nabla_1^{+} A_2(0, n) -\frac12 |\phi (0, n)|^2=0,  \quad n \in \mathbb{Z}.
\end{equation}
Assume that there exists a solution to the discrete CSS system \eqref{dcss} with the given initial data
in the sense of Definition \ref{def-sol}. Then, for every $t \in [-T,T]$, the solution satisfies
\begin{align}\label{conser2}
\nabla_1^{+} A_2(t, n) -\frac12 |\phi (t, n)|^2=0, \quad n \in \mathbb{Z}.
\end{align}
\end{lem}

\begin{proof}
We note the following identity:
\begin{align}
\label{alg}
\nabla_1^+ (\bar{\phi} D_1^+ \phi)(n) =  \overline{D_1^+ \phi}(n) D_1^+ \phi (n) +  \bar{\phi}(n+1) D_1^{-} D^{+}_1 \phi(n+1).
\end{align}
Using the first three equations of the system \eqref{dcss} and the identity \eqref{alg}, we obtain
\begin{align*}
\partial_t ( \nabla_1^{+} A_2 (t, n) -\frac12 |\phi (t, n)|^2 )
 &= - \nabla_1^{+}\mbox{Im} (\bar{\phi}(t, n-1) D_1^{+} \phi (t, n-1)) -\mbox{Im}(i \bar{\phi} \partial_t \phi) (t, n)\\
 &= - \mbox{Im} (\bar{\phi}(t, n) D_1^{-}D_1^{+} \phi(t, n))-\mbox{Im}(i \bar{\phi} D_0\phi)(t, n) \\
 &= - \mbox{Im} \left( \bar{\phi}(t, n)   (i D_0\phi+ D_1^{-}D_1^{+} \phi )(t, n) \right) \\
 & =0.
\end{align*}
Due to this conservation, the relation \eqref{conser2} remains true for every $t \in [-T,T]$ as long as
a solution to the discrete CSS system \eqref{dcss} with the given initial data satisfying the constraint \eqref{initial-id} exists
in the sense of Definition \ref{def-sol}.
\end{proof}

\begin{rem}
\label{rem-redundancy}
The third and fourth equations of the system \eqref{dcss} represents the balance equation
for the scalar field $\phi$. Indeed, eliminating $A_2$ by
$$
\partial_t \nabla_1^+ A_2(t,n) = \nabla_1^+ \partial_t A_2(t,n),
$$
we obtain
\begin{equation}
\label{bal-1}
\frac{1}{2} \partial_t |\phi(t,n)|^2 + \nabla_1^+ \mbox{\rm Im} (\bar{\phi}(t,n-1) D_1^+ \phi(t,n-1)) = 0,
\end{equation}
which follows from the first equation of the system \eqref{dcss} thanks to the identity \eqref{alg}.
\end{rem}

\begin{rem}
\label{rem-mass}
Summing up the balance equation \eqref{bal-1} in $n \in \mathbb{Z}$ yields
\begin{align}\label{conser1}
\frac{d}{dt} \sum_{n \in \mathbb{Z}} |\phi(t,n)|^2  =0,
\end{align}
for the solution $\phi \in C^1([-T,T],\ell^2(\mathbb{Z}))$. This implies conservation of mass $M$
given by \eqref{mass-discrete}. The conservation of mass $M$ in the discrete CSS system \eqref{dcss}
generalizes the conservation of mass $M$ given by \eqref{mass} for the continuous CSS system \eqref{css}.
However, the discrete CSS system \eqref{dcss} does not exhibit conservation of energy
which would be similar to the conservation of the total energy $E$ given by \eqref{energy} for the
continuous CSS system \eqref{css}.
\end{rem}

It follows from Lemma \ref{lem-redundancy} that the system \eqref{dcss} is under-determined since it consists of three
time evolution equations on four unknown fields, whereas the fourth equation represents a constrained preserved in
the time evolution. In order to close the system, we add a gauge condition, thanks to the gauge invariance
of the discrete CSS system \eqref{dcss} expressed by the following lemma.

\begin{lem}
\label{lem-gauge}
Assume there exists a solution to the discrete CSS system \eqref{dcss} in the sense of Definition $\ref{def-sol}$.
Let $\chi$ be an arbitrary function in $C^1([-T,T],\ell^{\infty}(\mathbb{Z}))$. Then,
the following gauge-transformed functions
\begin{align}
\label{gauge}
\left\{
\begin{aligned}
& \tilde{\phi}(t,\,n) = e^{i \chi(t,\,n)}\phi(t,\,n),\\
& \tilde{A_0} (t,\,n)= A_0 (t,\,n)+ \partial_t \chi (t,\,n),\\
& \tilde{A_1} (t,\,n+\frac12)= A_1 (t,\,n+\frac12)+ \nabla_1^{+} \chi (t,\,n),\\
& \tilde{A_2} (t,\,n)= A_2(t,\,n),\end{aligned} \right.
\end{align}
also provide a solution to the discrete CSS system \eqref{dcss} in the sense of Definition $\ref{def-sol}$.
\end{lem}

\begin{proof}
We proceed with the explicit computations:
\begin{align*}
D_0 \tilde{\phi} & = \partial_t \tilde{\phi}(t, \,n)-i \tilde{A_0}(t, \,n) \tilde{\phi}(t, \,n)\\
& = e^{i\chi (t, \,n)}\left(\partial_t \phi(t, \,n)-i A_0(t, \,n) \phi(t, \,n) \right) = e^{i \chi(t,n)} D_0 \phi,\\
D_1^+ \tilde{\phi} & = e^{-ih \tilde{ A_1}( t, \,n+\frac12)}\tilde{\phi}(t, \,n+ 1)-  \tilde{\phi}(t,\, n) \\
&= e^{i\chi (t, \,n)}\left( e^{-ih  A_1( t, \,n+\frac12)}\phi(t, \,n+ 1)-  \phi(t,\, n) \right) = e^{i \chi(t,n)} D_1^+ \phi,\\
D_1^- \tilde{\phi} & = \tilde{\phi}(t, \,n)-e^{ih \tilde{ A_1}( t, \,n-\frac12)}  \tilde{\phi}(t,\, n-1) \\
&= e^{i\chi (t, \,n)}\left( \phi(t,\,n)-e^{ih A_1(t,\, n-\frac12)}  \phi(t,\, n-1)\right) = e^{i \chi(t,n)} D_1^- \phi,
\end{align*}
and
\begin{align*}
\partial_t \tilde{A_1}(n+\frac12, t) -\nabla_1^{+} \tilde{A_0} (n, t) &= \partial_t A_1(n+\frac12, t) -\nabla_1^{+} A_0 (n, t),
\end{align*}
where we have used
$$
e^{-i h \nabla_1^+ \chi (t, n)}   e^{i \chi (t, n+1)}=e^{i \chi(t, n)} \quad \mbox{\rm and} \quad
e^{i h \nabla_1^+ \chi (t, n-1)}   e^{i \chi (t, n-1)}=e^{i \chi(t, n)}.
$$
Under the conditions of the lemma, $\partial_t \chi$ and $\nabla_1^+ \chi$ are defined in $\ell^{\infty}(\mathbb{Z})$
for every $t \in [-T,T]$. Thanks to the transformation above, both $(\phi,A_0,A_1,A_2)$ and
$(\tilde{\phi},\tilde{A}_0,\tilde{A}_1,\tilde{A}_2)$ are solutions of the same system \eqref{dcss}.
\end{proof}

\begin{rem}
If the standard difference method is used to express the continuous covariant derivative $D_1 \phi$ by
its discrete counterparts $D_1^{\pm} \phi$ in the form:
\[
D_1^{\pm}\phi = \nabla_1^{\pm} \phi -i A_1 \phi,
\]
then the resulting discrete CSS system is not gauge invariant.
This remark illustrates the importance of using the discrete covariant derivatives in the form \eqref{covariant}.
\end{rem}

It follows from the gauge transformation \eqref{gauge} of Lemma \ref{lem-gauge} that a solution to the discrete CSS system \eqref{dcss}
is formed by a class of gauge equivalent field $(\phi, \,A_0,\, A_1, \,A_2)$. Two types of gauge conditions
are typically considered to break the gauge symmetry: either $A_0(t,n) = 0$ by appropriate choice
of $\partial_t \chi$ or $A_1(t,n + \frac{1}{2}) = 0$ by appropriate choice of $\nabla_1^+ \chi$.

In the continuous Maxwell (Yang-Mills) or Chern-Simons gauge equations, the temporal gauge condition $A_0\equiv 0$ has been used by several authors
\cite{EM1, GV,  P}. In the space of $(1+1)$ dimensions, the spatial gauge condition  $A_1\equiv 0$ was used in \cite{Huh, Huh-d, Huh2}
to simplify the related system of  equations. Here in the discrete setting, we will use the spatial gauge condition
for the same purpose and set $A_1(t,n + \frac{1}{2}) = 0$.

The discrete CSS system \eqref{dcss} with $A_1 \equiv 0$ simplifies to the form:
\begin{align}\label{dcssr}
\left\{
 \begin{aligned}
& i \partial_t \phi (t, n) + A_0 (t, n) \phi (t, n) +\nabla_1^{-} \nabla_1^+ \phi(t, n) - A_2^2(t, n) \phi(t, n) \\
& \phantom{texttexttexttext} + \lambda |\phi(t, n)|^{2p}\,\phi(t, n)=0, \\
&  \nabla_1^{+} A_0 (t, n)=  A_2(t, n) |\phi(t, n)|^2,\\
& \nabla_1^{+} A_2 (t, n) =\frac12 |\phi (t, n)|^2,
 \end{aligned}\right.
\end{align}
where we removed the redundant time evolution equation for $A_2$ thanks to the results in
Lemma \ref{lem-redundancy} and Remark \ref{rem-redundancy}. We show well-posedness
of the initial-value problem for the coupled system \eqref{dcssr}, which yields
the proof of Theorem \ref{lem-WPS}.

\vspace{0.2cm}

\begin{pf1}{\it of Theorem \ref{lem-WPS}}.
By Lemmas \ref{lem-redundancy} and \ref{lem-gauge}, the constraints described in
the second and third equations of the system \eqref{dcssr} are preserved in the time
evolution of the first equation of the system \eqref{dcssr}. The initial data
$A_0(0,n) = \mathcal{A}_n$ and $A_2(0,n) = \mathcal{B}_n$ satisfy the consistency
conditions \eqref{consistency} which agree with the second and third equations of the system \eqref{dcssr}.

By inverting the difference operators under the boundary conditions \eqref{boundary-values},
we derive the closed-form solutions for $A_0$ and $A_2$:
 \begin{align*}
 A_2(t, n) = \beta - \frac{h}{2} \sum_{k=n}^{\infty} |\phi(t, k)|^2,   \quad
 A_0(t, n) = \alpha - h \sum_{k=n}^{\infty} A_2(t,k)|\phi(t, k)|^2,
 \end{align*}
which yield the bounds
 \begin{align}
 \label{imp-bounds}
& |A_2(t,n)| \leq |\beta| + \frac{h}{2} \|\phi(t,\cdot)\|^2_{\ell^2}, \\
& |A_0(t,n)| \leq |\alpha| +  h \left( |\beta| + \frac{h}{2} \|\phi(t,\cdot)\|^2_{\ell^2} \right)   \|\phi(t,\cdot)\|^2_{\ell^2}.
 \label{imp-bounds-new}
 \end{align}
Thanks to these bounds, the initial-value problem
for the system \eqref{dcssr} can be written as an integral
equation on $\phi$ in the space of continuous functions of time with range in $\ell^2(\mathbb{Z})$.

Local well-posedness on a small time interval $(-\tau,\tau)$ with $\tau > 0$
follows from the contraction mapping theorem (see, e.g., \cite{KLS})
thanks to the Banach algebra property of $\ell^2(\mathbb{Z})$,
bounds on the linear operator $\nabla_1^- \nabla_1^+$ in $\ell^2(\mathbb{Z})$
and bounds \eqref{imp-bounds}--\eqref{imp-bounds-new}. Global well-posedness in $\ell^2(\mathbb{Z})$  follows from
the mass conservation $\| \phi(t,\cdot) \|^2_{\ell^2} = \| \Phi \|^2_{\ell^2}$,
where $\phi(0,n) = \Phi_n$.
\end{pf1}

\section{Existence of stationary bound states}
\setcounter{equation}{0}
\label{sec-4}

Here we consider the existence of stationary bound states for the discrete CSS system \eqref{dcss} with
the gauge condition $A_1 \equiv 0$. In the end, we will prove Theorems \ref{thm-existence} and \ref{thm-nonexistence}.

The last two equations of the system \eqref{dcssr} allow us to reduce
$A_0$ and $A_2$ to only one variable $g(t,n) := A_0(t,n) - A_2^2(t,n)$ since
\begin{align*}
\nabla_1^{+}(A_0-A_2^2)&=A_2(t,n)|\phi(t,n)|^2-\frac12 |\phi(t,n)|^2 (A_2(t,n+1) + A_2(t,n))\\
&=-\frac{h}{2} |\phi(t,n)|^2 \nabla_1^{+}A_2(t,n) \\
&= -\frac{h}{4} |\phi(t,n)|^4.
\end{align*}
Thus, the system \eqref{dcssr} can be closed at the following system of two equations:
\begin{align}\label{our}
\left\{
 \begin{aligned}
& i \partial_t \phi(t,n)  + \nabla_1^{-} \nabla_1^+ \phi(t,n) + g(t,n) \phi(t,n) + \lambda |\phi(t,n)|^{2p} \,\phi(t, n)=0, \\
& \nabla_1^{+} g(t, n)= -\frac{h}{4} |\phi(t, n)|^4.
 \end{aligned}
 \right.
\end{align}

In the continuum limit $h \to 0$, the second equation of the system \eqref{our} yields
$\partial_x g = 0$, which is solved by $g = 0$ up to an arbitrary constant (see Remark \ref{rem-gamma}),
whereas the first equation of the system \eqref{our} yields formally the continuous NLS equation
\begin{align}\label{conti}
i \partial_t \varphi  +\partial_{x}^2 \varphi + \lambda |\varphi|^{2p}\varphi=0,
\end{align}
where $\varphi(t,x)$ is assumed to be a smooth function such that $\phi(t,n) = \varphi(t,hn)$, $n \in \Bbb{Z}$.
The continuous NLS equation \eqref{conti} also follows from integration of the continuous CSS system \eqref{css}
with the gauge condition $A_1\equiv 0$ (see \cite{Huh} for details).

The gauge field $g$ does not appear in the continuous NLS equation \eqref{conti}. The same continuous NLS equation
\eqref{conti} is also derived in the continuum limit $h \to 0$ of the standard discrete NLS equation:
\begin{equation}
\label{dnls}
 i \partial_t \phi(t,n)  + \nabla_1^{-} \nabla_1^+ \phi(t,n) + \lambda |\phi(t,n)|^{2p}\,\phi(t,n) = 0.
\end{equation}
The discrete NLS equation \eqref{dnls} was investigated in many recent studies (see, e.g., \cite{panos-book,Pelin-book}).
In particular, it admits a large set of stationary bound states, which includes the ground state of energy at
fixed mass \cite{Weinstein}. In the cubic case $p = 1$, the ground state exists for every $h > 0$ and converges
in the continuum limit $h \to 0$ to the single-humped solitary wave of the continuous NLS equation \eqref{conti}
and in the anti-continuum limit $h \to \infty$ to a single-site solution \cite{panos-book,Pelin-book}.
We will show that these properties of the ground state in the discrete NLS equation \eqref{dnls} are very different
from properties of the stationary bound states in the discrete CSS system \eqref{our}.

Substituting
\[
\phi(t,n) = e^{i \Omega t} U_n, \quad g(t,n) = G_n, \quad t \in \mathbb{R}, \quad n \in \mathbb{Z}
\]
into the discrete CSS system \eqref{our} yields the following system of difference equations for
sequences $U := \{ U_n \}_{n \in \mathbb{Z}}$ and $G := \{ G_n \}_{n \in \mathbb{Z}}$:
 \begin{align}\label{equ23}
 \left\{
 \begin{aligned}
&     \frac{1}{h^2} \left( U_{n+1}-2 U_{n}+ U_{n-1} \right)-\Omega U_n + G_n U_n+ \lambda |U_n|^{2p} U_n=0, \\
& G_{n+1}-G_n= -\frac{h^2}{4} |U_n|^4.
 \end{aligned}
\right.
\end{align}

\begin{rem}
There are two critical exponents $p = 1$ and $p = 2$ in the system \eqref{equ23}.
For $p = 1$, the lattice spacing parameter $h$ can be scaled out thanks to the scaling transformation:
\begin{equation}
\label{scaling-cubic}
p = 1: \quad \tilde{U} = h U, \quad \tilde{G} = h^2 G, \quad \tilde{\Omega} = h^2 \Omega,
\end{equation}
where $\tilde{U}$, $\tilde{G}$, and $\tilde{\Omega}$ solve the same system \eqref{equ23} but with $h = 1$.
For $p = 2$, the nonlinear terms in the two equations of the system \eqref{equ23} have the same
exponents.
\label{rem-critical}
\end{rem}

In order to prove persistence of single-site solutions in the anti-continuum limit $h \to \infty$,
we close the system \eqref{equ23} with the following relation:
\begin{equation}
\label{tilde-G}
G_n = \gamma + \frac{h^2}{4} \sum_{k=n}^{\infty} |U_k|^4,
\end{equation}
where $\gamma := \lim_{n \to \infty} G_n$ is another parameter and $U \in \ell^4(\mathbb{Z})$ is assumed.

\begin{rem}
Parameter $\gamma$ in \eqref{tilde-G} can be set to $0$ without loss of generality thanks to the
transformation
$$
\phi(t,n) \mapsto \tilde{\phi}(t,n) = \phi(t,n) e^{i \gamma t}
$$
between two solutions to the discrete CSS system \eqref{our}.
\label{rem-gamma}
\end{rem}

By Remark \ref{rem-gamma}, we set $\gamma = 0$ and substitute \eqref{tilde-G} into the first equation
of the system \eqref{equ23}. This yields the root-finding problem $F(U,h) = 0$, where
 \begin{align}\label{root}
[F(U,h)]_n := \left( \lambda |U_n|^{2p} - \Omega + \frac{h^2}{4} \sum_{k=n}^{\infty} |U_k|^4 \right) U_n
+ \frac{1}{h^2} \left( U_{n+1}-2 U_{n} + U_{n-1} \right), \quad n \in \mathbb{Z}.
\end{align}
For the proof of Theorem \ref{thm-existence},
parameters $\lambda$, $\Omega$,  and $p$ are fixed, whereas $h$ is considered to be large.
The following lemma shows that the vector field in \eqref{root} is closed if $U \in \ell^1(\mathbb{Z})$.

\begin{lem}
\label{lem-1}
The mapping $F(U,h) : \ell^1(\mathbb{Z}) \times \mathbb{R}^+ \mapsto \ell^1(\mathbb{Z})$
is $C^1$ if $p > 0$.
\end{lem}

\begin{proof}
The discrete Laplacian is a bounded operator as
$$
\| \nabla_1^- \nabla_1^+ U \|_{\ell^1} \leq \frac{4}{h^2} \| U \|_{\ell^1},
$$
whereas the nonlinear term is closed in $\ell^1(\mathbb{Z})$ thanks to the continuous embeddings
of $\ell^1(\mathbb{Z})$ to $\ell^4(\mathbb{Z})$ and the elementary inequality
$$
\sum_{n \in \mathbb{Z}} \sum_{k=n}^{\infty} |U_k|^4 |U_n| = \sum_{k \in \mathbb{Z}} |U_k|^4 \sum_{n = -\infty}^{k} |U_n|
\leq \| U \|_{\ell^1} \| U \|_{\ell^4}^4 \leq \| U \|_{\ell^1}^5.
$$
The mapping $F(U,h) : \ell^1(\mathbb{Z}) \times \mathbb{R}^+ \mapsto \ell^1(\mathbb{Z})$ is closed
and locally bounded. It depends on powers of $U$ and linear terms of $h^2$ and $1/h^2$. Therefore, it is $C^1$
for every $U \in \ell^1(\mathbb{Z})$ and $h \in \mathbb{R}^+$.
\end{proof}

The local part of $F(U,h)$ leads to the root-finding equation
\begin{equation}
\label{root-again}
\lambda \mathcal{U}^{2p} - \Omega + \frac{h^2}{4} \mathcal{U}^4 = 0,
\end{equation}
for which we are only interested in the positive roots for $\mathcal{U}$. The following lemma
controls uniqueness and the asymptotic expansion of the positive roots of the nonlinear equation \eqref{root-again} as $h \to \infty$.

\begin{lem}
\label{lem-0}
Fix $\Omega > 0$, $\lambda > 0$, and $p > 0$. For every $h > 0$, there is only one positive root
of the nonlinear equation \eqref{root-again} labeled as $\mathcal{U}(h)$. Moreover,
\begin{equation}
\label{root-asymptotic}
\mathcal{U}(h) = \frac{\sqrt[4]{4 \Omega}}{\sqrt{h}} \left[ 1 + \mathcal{O}(h^{-p}) \right] \quad \mbox{\rm as} \quad h \to \infty.
\end{equation}
\end{lem}

\begin{proof}
Since the function $\mathbb{R}^+ \ni x \mapsto \lambda x^{2p} + \frac{h^2}{4} x^2 \in \mathbb{R}^+$ is monotonically increasing,
there is exactly one intersection of its graph with the level $\Omega > 0$. Therefore, there exists only one positive root
of the nonlinear equation \eqref{root-again} labeled as $\mathcal{U}(h)$. By using scaling $\mathcal{U} = \mathcal{V} h^{-1/2}$,
we transform the nonlinear equation \eqref{root-again} to the equivalent form $f(\mathcal{V},\textsf{h}) = 0$, where
\begin{equation}
\label{root-again-V}
f(\mathcal{V},\textsf{h}) := \frac{1}{4} \mathcal{V}^4 - \Omega + \lambda \textsf{h} \mathcal{V}^{2p}, \quad
\textsf{h} := h^{-p}.
\end{equation}
Let $\mathcal{V}_0 := \sqrt[4]{4 \Omega}$ be the root of $f(\mathcal{V}_0,0) = 0$. Since $f$ is $C^1$ with respect to $\mathcal{V}$
and linear with respect to $\textsf{h}$ with $\partial_{\mathcal{V}} f(\mathcal{V}_0,0) = \mathcal{V}_0^3 > 0$,
the implicit function theorem implies the existence and uniqueness of the root
$\mathcal{V}(\textsf{h})$ of the nonlinear equation $f(\mathcal{V}(\textsf{h}),\textsf{h}) = 0$ for every small $\textsf{h}$ such that
$\mathcal{V}$ is $C^1$ with respect to $\textsf{h}$ and $\mathcal{V}(\textsf{h}) = \mathcal{V}_0 + \mathcal{O}(\textsf{h})$ as $\textsf{h} \to 0$.
By uniqueness of the positive root $\mathcal{U}(h)$, we obtain $\mathcal{U}(h) = \mathcal{V}(\textsf{h}) h^{-1/2}$,
which yields the asymptotic expansion \eqref{root-asymptotic}.
\end{proof}

By Lemma \ref{lem-0}, we set $U = V h^{-1/2}$ and rewrite the root-finding problem \eqref{root} in the equivalent
form $\textsf{F}(V,\textsf{h},\varepsilon) = 0$, where
 \begin{align}\label{root-new}
[\textsf{F}(V,\textsf{h},\varepsilon)]_n := \left( \lambda \textsf{h} |V_n|^{2p} - \Omega + \frac{1}{4} \sum_{k=n}^{\infty} |V_k|^4 \right) V_n
+ \varepsilon \left( V_{n+1}-2 V_{n} + V_{n-1} \right), \quad n \in \mathbb{Z},
\end{align}
with $\textsf{h} := h^{-p}$ and $\varepsilon := h^{-2}$. The following lemma shows that the limiting configuration
$V = \mathcal{V}(\textsf{h}) \delta_0$ with $\delta_0$ being Kronecker's delta function given by \eqref{delta}
persists with respect to small parameter $\varepsilon$ for any small $\textsf{h}$.

\begin{lem}
\label{lem-2}
Fix $\Omega > 0$, $\lambda > 0$, and $p > 0$. For every $\textsf{h} > 0$, there exists $\varepsilon_0 > 0$
such that for every $\varepsilon \in (0,\varepsilon_0)$ the root-finding
problem $\textsf{F}(V,\textsf{h},\varepsilon) = 0$ with $\textsf{F}$ given by \eqref{root-new}
admits the unique solution $V(\textsf{h},\varepsilon) \in \ell^1(\mathbb{Z})$
such that $V(\textsf{h},\varepsilon)$ is $C^1$ with respect to $\varepsilon$ and
\begin{equation*}
V(\textsf{h},\varepsilon) = \mathcal{V}(\textsf{h}) \delta_0 + \mathcal{O}(\varepsilon) \quad \mbox{\rm as} \quad \varepsilon \to 0,
\end{equation*}
where $\mathcal{V}(\textsf{h}) = h^{1/2} \mathcal{U}(h)$ is defined by Lemma $\ref{lem-0}$.
\end{lem}

\begin{proof}
We check the three conditions of the implicit function theorem.
The mapping $\textsf{F}(V,\textsf{h},\varepsilon) : \ell^1(\mathbb{Z}) \times
\mathbb{R}^+ \times \mathbb{R}^+ \mapsto \ell^1(\mathbb{Z})$
is $C^1$ by Lemma \ref{lem-1}. By Lemma \ref{lem-0}, we have
$\textsf{F}(\mathcal{V}(\textsf{h}) \delta_0,\textsf{h},0) = 0$.
Finally, we compute the Jacobian of $\textsf{F}(V,\textsf{h},\varepsilon)$
at $(\mathcal{V}(\textsf{h}) \delta_0, \textsf{h}, 0)$, which is a diagonal operator
with the diagonal entries:
$$
[D_V \textsf{F}(\mathcal{V}(\textsf{h}) \delta_0,\textsf{h},0)]_{nn} =
\left\{ \begin{array}{ll} \frac{1}{4} \mathcal{V}(\textsf{h})^4 - \Omega, \quad & n \in \mathbb{Z}^-, \\
(2p+1) \lambda \textsf{h} \mathcal{V}(\textsf{h})^{2p} + \frac{5}{4} \mathcal{V}(\textsf{h})^4 - \Omega, \quad & n = 0, \\
- \Omega, \quad & n \in \mathbb{Z}^+, \end{array} \right.
$$
where $\mathbb{Z}^{\pm} := \{ \pm 1, \pm 2, \dots \}$.
With the account of the nonlinear equation $f(\mathcal{V},\textsf{h}) = 0$ with $f$ given by
\eqref{root-again-V}, the Jacobian operator can be rewritten in the form:
\begin{equation}
\label{Jacobian-explicit}
[D_V \textsf{F}(\mathcal{V}(\textsf{h}) \delta_0,\textsf{h},0)]_{nn} =
\left\{ \begin{array}{ll} -\lambda \textsf{h} \mathcal{V}(\textsf{h})^{2p}, \quad & n \in \mathbb{Z}^-, \\
2p \lambda \textsf{h} \mathcal{V}(\textsf{h})^{2p} + \mathcal{V}(\textsf{h})^4, \quad & n = 0, \\
- \Omega, \quad & n \in \mathbb{Z}^+. \end{array} \right.
\end{equation}
For every $\textsf{h} > 0$, the Jacobian operator $D_V \textsf{F}(\mathcal{V}(\textsf{h}) \delta_0,\textsf{h},0)$
is invertible so that the assertion of the lemma follows by the implicit function theorem.
\end{proof}

\vspace{0.2cm}

\begin{pf1}{\it of Theorem \ref{thm-existence}.}
In order to apply the result of Lemma \ref{lem-2} to the root-finding problem $F(U,h) = 0$ with $F$ given by \eqref{root},
we should realize that both parameters $\textsf{h} = h^{-p}$ and $\varepsilon = h^{-2}$ are small as $h \to \infty$.
As a result, the Jacobian operator $D_V \textsf{F}(\mathcal{V}(\textsf{h}) \delta_0,\textsf{h},0)$
given by \eqref{Jacobian-explicit} becomes singular
in the limit $h \to \infty$. To be precise, it follows from the explicit expression \eqref{Jacobian-explicit}
that there exists a positive constant $C$ independent of $h$ such that
$$
\| D_V \textsf{F}(\mathcal{V}(\textsf{h}) \delta_0,\textsf{h},0) \|_{\ell^1(\mathbb{Z}) \to \ell^1(\mathbb{Z})} \geq C h^{-p}.
$$

In order to show that the root of $\textsf{F}(V,h^{-p},h^{-2})$ exists for large $h$ for every $p \in (0,2)$, we rewrite
the system $\textsf{F}(V,h^{-p},h^{-2}) = 0$ componentwise:
\begin{eqnarray}
\label{system-1}
& n \in \mathbb{Z}^-  : & \;
\left( \lambda h^{-p} |V_n|^{2p} - \Omega + \frac{1}{4} \sum_{k=n}^{\infty} |V_k|^4 \right) V_n
+ h^{-2} \left( V_{n+1}-2 V_{n} + V_{n-1} \right) = 0, \\
\label{system-2}
& n = 0 : & \;
\left( \lambda h^{-p} |V_0|^{2p} - \Omega + \frac{1}{4} \sum_{k=0}^{\infty} |V_k|^4 \right) V_0
+ h^{-2} \left( V_{1}-2 V_0 + V_{-1} \right) = 0, \\
\label{system-3}
& n \in \mathbb{Z}^+ : & \;
\left( \lambda h^{-p} |V_n|^{2p} - \Omega + \frac{1}{4} \sum_{k=n}^{\infty} |V_k|^4 \right) V_n
+ h^{-2} \left( V_{n+1}-2 V_{n} + V_{n-1} \right) = 0.
\end{eqnarray}
By the last line in \eqref{Jacobian-explicit}, the Jacobian operator for system \eqref{system-3}
is invertible in $\ell^1(\mathbb{Z}^+)$ and the inverse operator is uniformly bounded as $h \to \infty$
if $\Omega > 1$ is fixed.
By the implicit function theorem, for every $V_0 \in \mathbb{R}$ and every large $h$, there exists the unique solution
$\{ V_n \}_{n \in \mathbb{Z}^+} \in \ell^1(\mathbb{Z}^+)$ to system \eqref{system-3} such that
$\| V \|_{\ell^1(\mathbb{Z}^+)} \leq C h^{-2} |V_0|$ for some positive $h$-independent constant $C$.

Similarly, because $\mathcal{V}(\textsf{h}) = \mathcal{V}_0 + \mathcal{O}(\textsf{h})$ as $h \to \infty$,
the middle line in \eqref{Jacobian-explicit} shows that if the solution
$\{ V_n \}_{n \in \mathbb{Z}^+} \in \ell^1(\mathbb{Z}^+)$ to system \eqref{system-3}
is substituted into \eqref{system-2}, then for every $V_{-1} \in \mathbb{R}$ and every large $h$,
there exists the unique solution $V_0 \in \mathbb{R}$ to system \eqref{system-2} such that
$|V_0 - \mathcal{V}(\textsf{h})| \leq C h^{-2} |V_{-1}|$ for some positive $h$-independent constant $C$.

Finally, we treat the remaining system \eqref{system-1}, for which $V_0$ and $\{ V_n \}_{n \in \mathbb{Z}^+}$ are expressed
from the unique solution to systems \eqref{system-2} and \eqref{system-3}. Thanks to the positivity of $|V_k|^4$
and the first line in \eqref{Jacobian-explicit}, we obtain for $n \in \mathbb{Z}^-$:
$$
\lambda h^{-p} |V_n|^{2p} - \Omega + \frac{1}{4} \sum_{k=n}^{\infty} |V_k|^4 \geq
-\Omega + \frac{1}{4} V_0^4 \geq C h^{-p},
$$
for some positive $h$-independent constant $C$. By the implicit function theorem, for every large $h$,
there exists the unique solution $\{ V_n \}_{n \in \mathbb{Z}^-} \in \ell^1(\mathbb{Z}^-)$ such that
$$
\| V \|_{\ell^1(\mathbb{Z}^-)} \leq C h^{p-2} |\mathcal{V}(\textsf{h})| \leq C h^{p-2}
$$
for some positive $h$-independent constant $C$. Since $p < 2$, we have $h^{p-2} \to 0$ as $h \to \infty$.

Combining all bounds together yields the unique root
$U = V h^{-1/2}$ to $F(U,h) = 0$. Recalling that $\mathcal{U}(h) = \mathcal{V}(\textsf{h}) h^{-1/2}$, we obtain the assertion of
Theorem \ref{thm-existence}.
\end{pf1}

\begin{rem}
The arguments based on Lemma $\ref{lem-2}$
are not sufficient for the proof of persistence of single-site solutions for $p \geq 2$ as $h \to \infty$.
This agrees with Remark \ref{rem-critical} since $p = 2$ is a critical power for the system \eqref{equ23}.
On the other hand, the critical power $p = 1$ in Remark $\ref{rem-critical}$ does not play any role if $\Omega$ is
fixed because the scaling transformation \eqref{scaling-cubic} which scales $h$ to unity requires us to scale
the parameter $\Omega$ in $h$.
\end{rem}

\begin{rem}
Besides the single-site solutions, other multi-site solutions can be considered in the anti-continuum limit
$h \to \infty$. In particular, the double-site solution is given by
\begin{equation}
\label{two-site}
U = \mathcal{W}(h) \delta_{0} + \mathcal{U}(h) \delta_{1},
\end{equation}
where $\mathcal{U}(h)$ is the same root of the nonlinear equation \eqref{root-again}, whereas
$\mathcal{W}(h)$ is a positive root of the following nonlinear equation:
\begin{equation}
\label{two-site-root1}
\lambda \mathcal{W}^{2p} - \Omega + \frac{h^2}{4} \mathcal{W}^4 + \frac{h^2}{4} \mathcal{U}^4 = 0,
\end{equation}
or equivalently,
\begin{equation}
\label{two-site-root2}
\lambda \mathcal{W}^{2p} + \frac{h^2}{4} \mathcal{W}^4 -\lambda \mathcal{U}^{2p} = 0.
\end{equation}
By the same arguments as in Lemma $\ref{lem-0}$, existence of the unique root $\mathcal{W}(h)$ can be proven
and the persistence analysis of Lemma $\ref{lem-2}$ holds verbatim for the double-site solution.
\label{rem-two-site}
\end{rem}

Finally, we give a proof of Theorem \ref{thm-nonexistence}, which relies on analysis of the root-finding problem
$F(U,h) = 0$ with $F$ given by \eqref{root} in the continuum limit $h \to 0$.

\vspace{0.2cm}

\begin{pf1}{\it of Theorem \ref{thm-nonexistence}.}
Let us rewrite the root-finding problem $F(U,h) = 0$ with $F$ given by \eqref{root}
in the following form:
 \begin{align}\label{scalar}
\frac{1}{h^2} \left( U_{n+1}-2 U_{n}+ U_{n-1} \right)-\Omega U_n + \frac{h^2}{4}
\left( \sum_{k=n}^{\infty} |U_k|^4 \right) U_n + \lambda |U_n|^{2p} U_n=0, \quad n \in \mathbb{Z}.
\end{align}
As previously, we assume that $U$ is real. Multiplying \eqref{scalar}
by $(U_{n+1}-U_{n-1})$ and summing in $n \in \mathbb{Z}$ under the same assumption
$U \in \ell^1(\mathbb{Z})$ yields the constraint
 \begin{align}\label{constraint}
\frac{h^2}{4} \sum_{n \in \mathbb{Z}} U_n^5 U_{n+1} + \lambda \sum_{n\in \mathbb{Z}} U_{n+1} U_n (U_n^{2p} - U_{n+1}^{2p}) = 0.
\end{align}
If $U_{-n} = U_n$, $n \in \mathbb{Z}$, then it follows directly that
$\sum_{n\in \mathbb{Z}} U_{n+1} U_n (U_n^{2p} - U_{n+1}^{2p}) = 0$, therefore,
the constraint \eqref{constraint} on existence of $U \in \ell^1(\mathbb{Z})$ reduces to
 \begin{align}\label{constraint-new}
\sum_{n \in \mathbb{Z}} U_n^5 U_{n+1} = 0.
\end{align}
We show that this constraint cannot be satisfied if $U$ satisfies the first bound in \eqref{error-bound} with $Q$
being the continuous NLS soliton in the exact form given by \eqref{soliton}. Indeed, we have
 \begin{align}\label{constraint-expanded}
\sum_{n \in \mathbb{Z}} U_n^5 U_{n+1} = \sum_{n \in \mathbb{Z}} Q(hn)^5 Q(hn + h) + \mathcal{R}(h),
\end{align}
where the residual term $\mathcal{R}(h)$ satisfies the bound
$$
|\mathcal{R}(h)| \leq C \left( \| Q(h \cdot) \|_{\ell^{\infty}}^5 + \| U \|_{\ell^{\infty}}^5 \right)
\| U - Q(h \cdot) \|_{\ell^1} \leq C h,
$$
since the embedding of $\ell^1$ into $\ell^{\infty}$ and the triangle inequality implies
$$
\| U \|_{\ell^{\infty}} \leq \| Q(h\cdot) \|_{\ell^{\infty}} + \| U - Q(h \cdot) \|_{\ell^{\infty}} \leq C,
$$
where the positive constant $C$ is $h$-independent and can change from one line to another line. Because $Q$ is $C^{\infty}$,
we use Riemann sums for smooth functions and rewrite the first term in \eqref{constraint-expanded} in the form
\begin{align}
\nonumber
\sum_{n \in \mathbb{Z}} Q(hn)^5 Q(hn + h) & =
\sum_{n \in \mathbb{Z}} Q(hn)^5 \left[ Q(hn) + h Q'(hn + \xi_n) \right] \\
& = \frac{1}{h} \left[ \int_{-\infty}^{\infty} Q(x)^6 dx + C h \right],
\label{constraint-expanded-new}
\end{align}
where $\xi_n \in [hn,hn+h]$ and $C > 0$ is $h$-independent. Since $\int_{\mathbb{R}} Q^6 dx > 0$ is $h$-independent,
it follows from \eqref{constraint-expanded} and \eqref{constraint-expanded-new} that
$\sum_{n \in \mathbb{Z}} U_n^5 U_{n+1} > 0$ for small $h$ and, hence, the constraint \eqref{constraint-new} cannot be satisfied.
This contradiction proves the assertion of the theorem. Finally, we note from \eqref{tilde-G} with $\gamma = 0$
by using the same estimates like in \eqref{constraint-expanded} and \eqref{constraint-expanded-new} that
$$
|G_n| \leq \frac{h^2}{4} \sum_{k \in \Bbb{Z}} U_k^4 \leq \frac{h^2}{4} \left[ \sum_{k \in \Bbb{Z}} Q(hk)^4 + Ch \right]
\leq \frac{h}{4} \left[ \int_{-\infty}^{\infty} Q(x)^4 dx + Ch \right],
$$
hence, the second bound in \eqref{error-bound} is implied by the first bound in \eqref{fig-0}.
\end{pf1}

\begin{rem}
The argument in the proof of Theorem $\ref{thm-nonexistence}$ does not
eliminate solutions of the system \eqref{equ23} for small $h$ which
are not close to the continuous NLS solitons in the sense of the bound \eqref{error-bound}.
\end{rem}

\section{Numerical results}
\setcounter{equation}{0}
\label{sec-5}

We approximate solutions of the difference equations \eqref{equ23} numerically by using the Newton--Raphson iteration algorithm
for the root-finding problem $F(U,h) = 0$, where $F(U,h)$ is given by \eqref{root}. The starting guess of the
iterative algorithm is either the single-site solution $\mathcal{U}(h) \delta_0$ or
the double-site solution $\mathcal{W}(h) \delta_0 + \mathcal{U}(h) \delta_1$, where $\mathcal{U}(h)$ and $\mathcal{W}(h)$ are found numerically from
the roots of the nonlinear equations \eqref{root-again} and \eqref{two-site-root2}. If iterations of the Newton--Raphson algorithm converge at
one value of $h$, we use the final approximation at this value of $h$ as a starting approximation for another
value of $h$ nearby. This parameter continuation is carried towards both the anti-continuum limit $h \to \infty$ and the continuum limit
$h \to 0$. We fix $\lambda = 1$ and use different values of parameter $\Omega > 0$ and $p > 0$ for such continuations in $h$.

Figure \ref{fig-0} gives examples of two stationary bound states of the system \eqref{equ23} for fixed $\Omega = 1$, $p = 1$, and $h = 3$.
One state is obtained by iterations from the single-site solution $\mathcal{U}(h) \delta_0$ (left panel),
whereas the other state is obtained by iterations from the double-site solution $\mathcal{W}(h) \delta_0 + \mathcal{U}(h) \delta_1$
(right panel).

\begin{figure}[h!]
	\centering
		\includegraphics[width=0.45\linewidth]{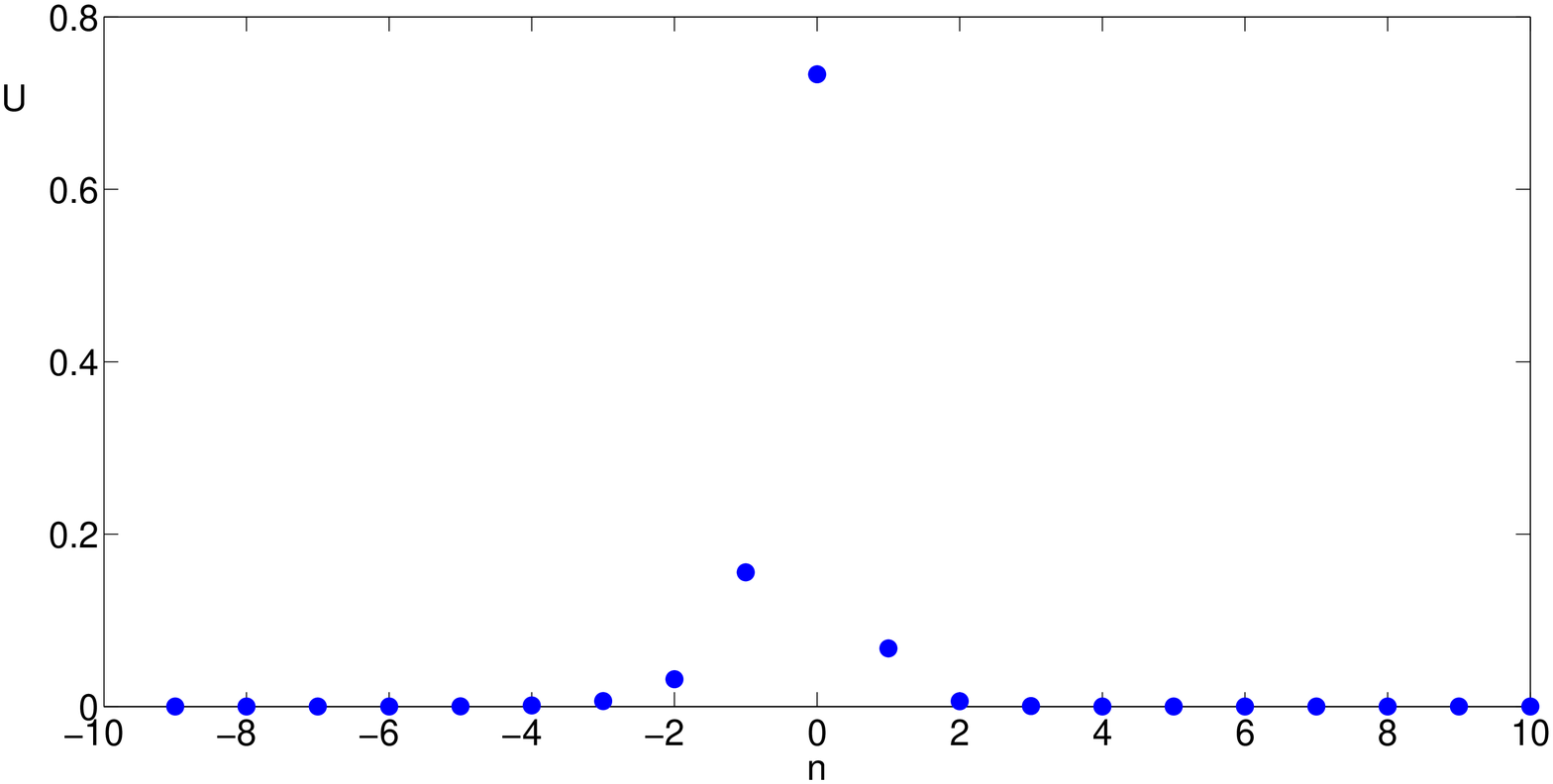}
		\includegraphics[width=0.45\linewidth]{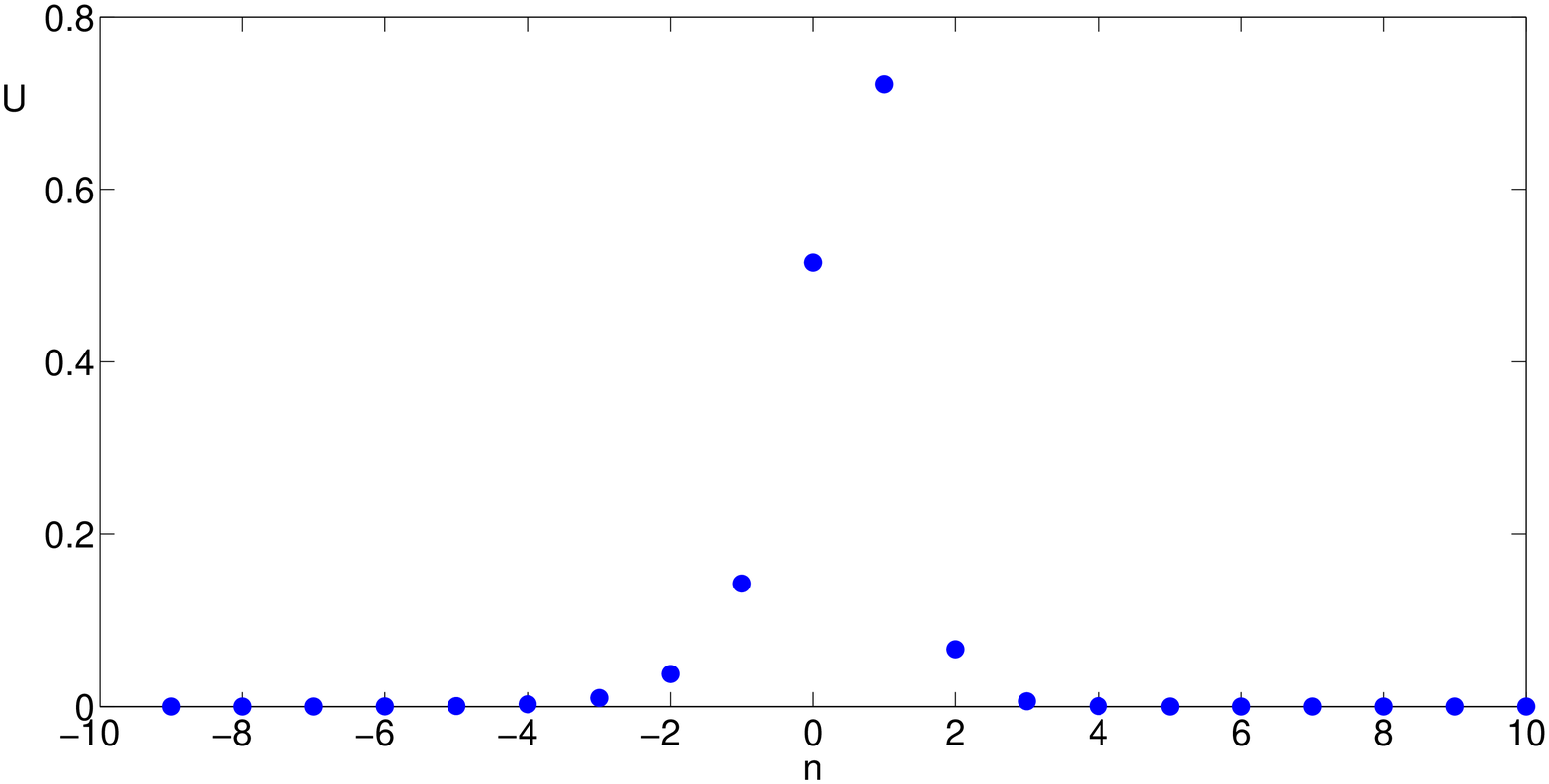}
\caption{Examples of single-site (left) and double-site (right) solutions $U$ for $p = 1$, $\Omega = 1$, and $h = 3$.}
\label{fig-0}
\end{figure}

Figure \ref{fig-1} shows the mass $M$ given by \eqref{mass-discrete} for the same two stationary states of the
system \eqref{equ23} versus $h$ for fixed $\Omega = 1$ with $p = 1$ (left) and $p = 3/2$ (right).
The lower branch corresponds to the single-site solution $\mathcal{U}(h) \delta_0$,
whereas the upper branch corresponds to the double-site solution $\mathcal{W}(h) \delta_0 + \mathcal{U}(h) \delta_1$.
By Theorem \ref{thm-existence} and Remark \ref{rem-two-site}, both branches of stationary states
extend to the anti-continuum limit of $h \to \infty$, as is confirmed in Fig. \ref{fig-1}.
On the other hand, in accordance with Theorem \ref{thm-nonexistence},
both branches do not extend to the continuum limit $h \to 0$ but coalesce in a fold bifurcation
at a critical value of $h$.

\begin{figure}[h!]
	\centering
		\includegraphics[width=0.45\linewidth]{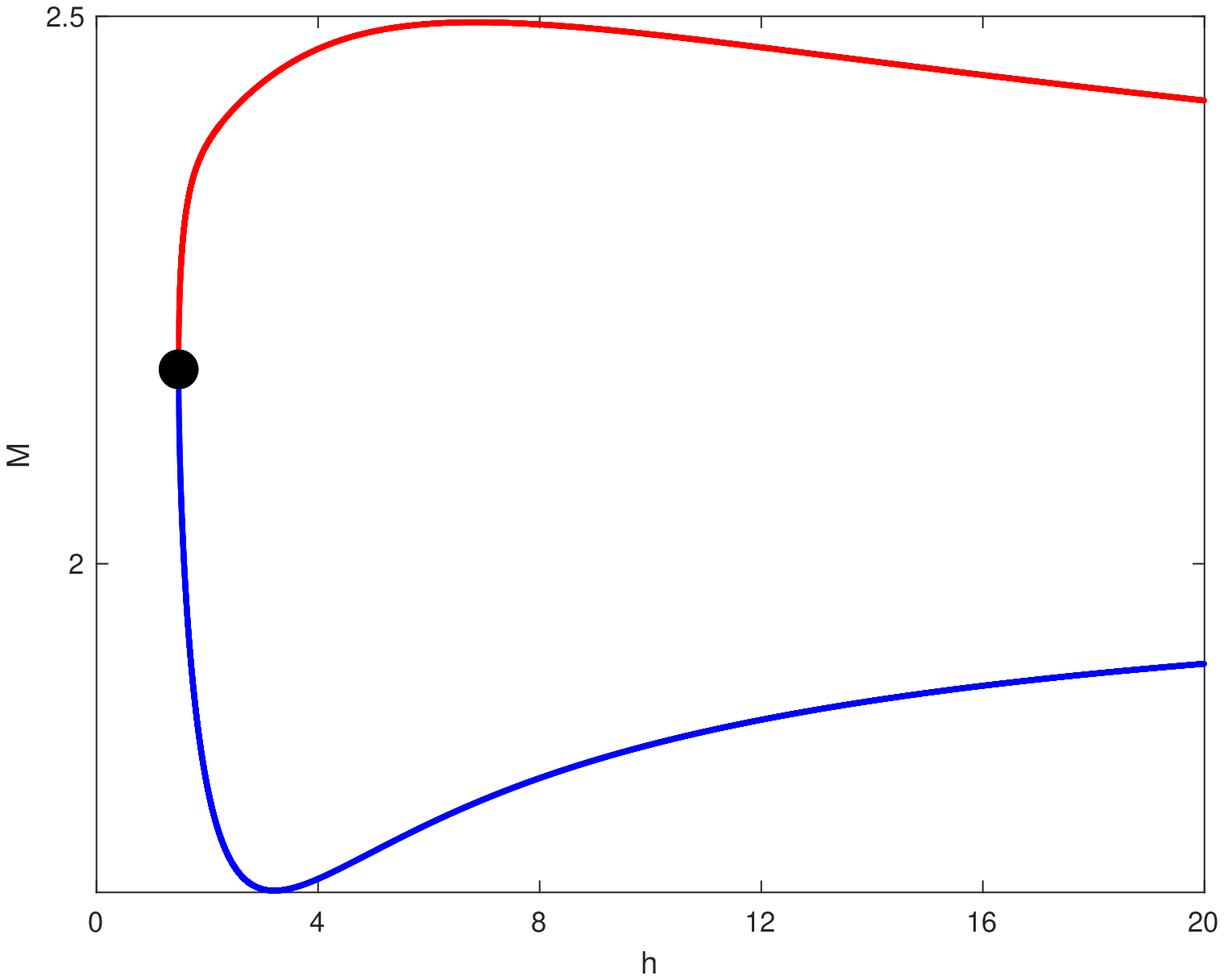}
		\includegraphics[width=0.45\linewidth]{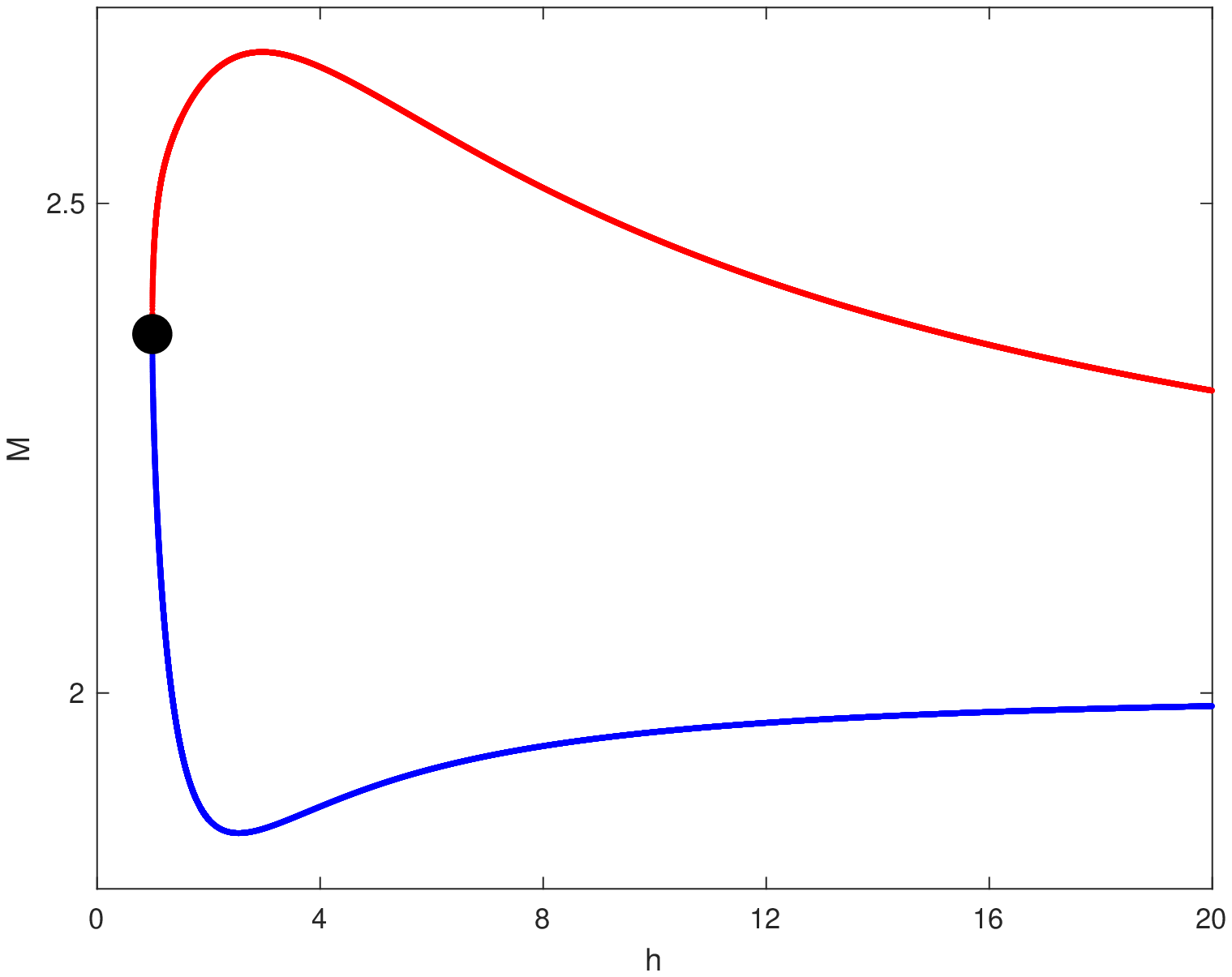}
	\caption{Mass $M$ in \eqref{mass-discrete} for two stationary states of system \eqref{equ23} versus
$h$ for $\lambda = 1$, $\Omega = 1$ and either $p = 1$ (left) or $p = 3/2$ (right). Blue curve shows
the stationary states obtained from the single-site solution $\mathcal{U}(h) \delta_0$, whereas red
curve shows the stationary states obtained from the double-site solution $\mathcal{W}(h) \delta_0 + \mathcal{U}(h) \delta_1$.
The big dot marks the fold bifurcations of the two branches. }
\label{fig-1}
\end{figure}

Figure \ref{fig-2} shows the same as Fig. \ref{fig-1} but for two values of $\Omega$ with $\Omega = 1$ and $\Omega = 10$.
Let $h_*(\Omega)$ denote the critical value of $h$ for the fold bifurcation.
It follows from Fig. \ref{fig-2} that the value of $h_*(\Omega)$ decreases with large values of $\Omega$.
For $p = 1$, this result follows from the scaling transformation \eqref{scaling-cubic}.
Since $\tilde{U}$, $\tilde{G}$, and $\tilde{\Omega}$ solve the same system \eqref{equ23} but with $h = 1$,
the fold bifurcation happens for some fixed value of $\tilde{\Omega}$ denotes by $\tilde{\Omega}_*$.
Then, for fixed $\Omega$, the value $h_*(\Omega)$ is found from the scaling transformation as
$$
h_*(\Omega) = \frac{\sqrt{\tilde{\Omega}_*}}{\sqrt{\Omega}},
$$
so that if $\Omega$ increases, then $h_*(\Omega)$ decreases.

For other values of $p$, a similar explanation can be provided based on
the generalized scaling transformation with parameter $a \in \mathbb{R}$:
\begin{equation}
\label{scaling}
\tilde{U} = h^{a} U, \quad \tilde{G} = h^{4a-2} G, \quad \tilde{\Omega} = h^{2ap} \Omega
\end{equation}
which reduces the system of difference equations \eqref{equ23} to the equivalent form:
 \begin{align}\label{eq-scaled}
 \left\{
 \begin{aligned}
& h^{-2(1-ap)} \left( \tilde{U}_{n+1}-2 \tilde{U}_{n}+ \tilde{U}_{n-1} \right) -
\tilde{\Omega} \tilde{U}_n + h^{2+2a(p-2)} \tilde{G}_n \tilde{U}_n + \lambda |\tilde{U}_n|^{2p} \tilde{U}_n = 0, \\
& \tilde{G}_{n+1} - \tilde{G}_n= -\frac{1}{4} |\tilde{U}_n|^4.
 \end{aligned}
\right.
\end{align}
If $p \neq 2$, the critical scaling between the two nonlinear terms occurs at $a = \frac{1}{2-p}$,
for which the first equation of the system \eqref{eq-scaled} can be rewritten in the form:
\begin{equation}
\label{eq-final-1}
h^{\frac{4(p-1)}{(2-p)}} \left( \tilde{U}_{n+1}-2 \tilde{U}_{n}+ \tilde{U}_{n-1} \right) -
\tilde{\Omega} \tilde{U}_n + \tilde{G}_n \tilde{U}_n + \lambda |\tilde{U}_n|^{2p} \tilde{U}_n = 0.
\end{equation}
If $p \in (1,2)$, then $h^{\frac{4(p-1)}{(2-p)}} \to 0$ as $h \to 0$. Let $\tilde{\Omega}_*(h)$ be
the value of $\tilde{\Omega}$ at the fold bifurcation in \eqref{eq-final-1} that depends on $h$. If we assume that
$\tilde{\Omega}_*(h)$ converges as $h \to 0$ to a nonzero value $\tilde{\Omega}_*(0)$,
then we have
$$
h_*(\Omega)^{\frac{2p}{2-p}} \approx \frac{\tilde{\Omega}_*(0)}{\Omega}
$$
so that if $\Omega$ increases, then $h_*(\Omega)$ decreases.

\begin{figure}[h!]
	\centering
		\includegraphics[width=0.45\linewidth]{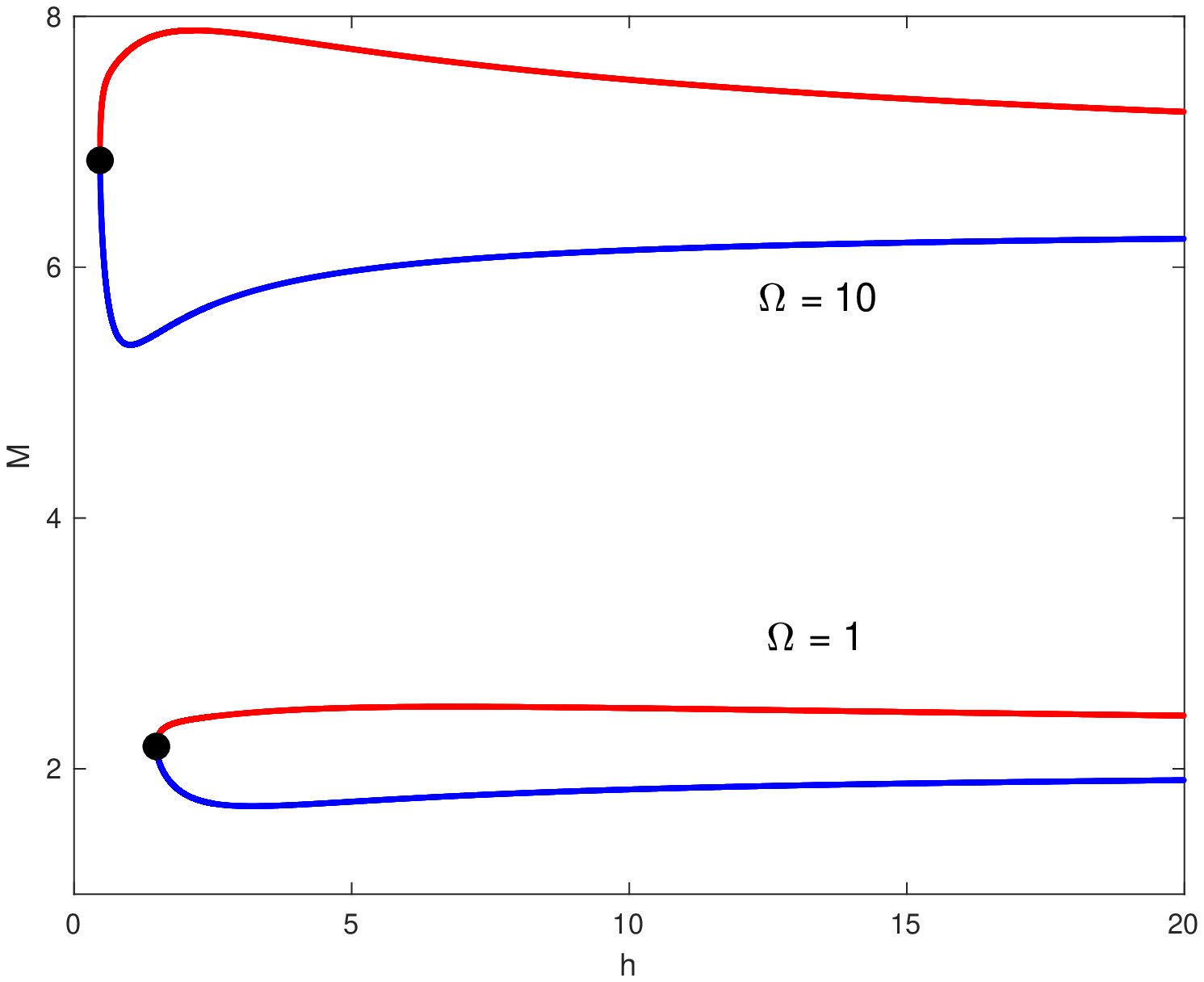}
		\includegraphics[width=0.45\linewidth]{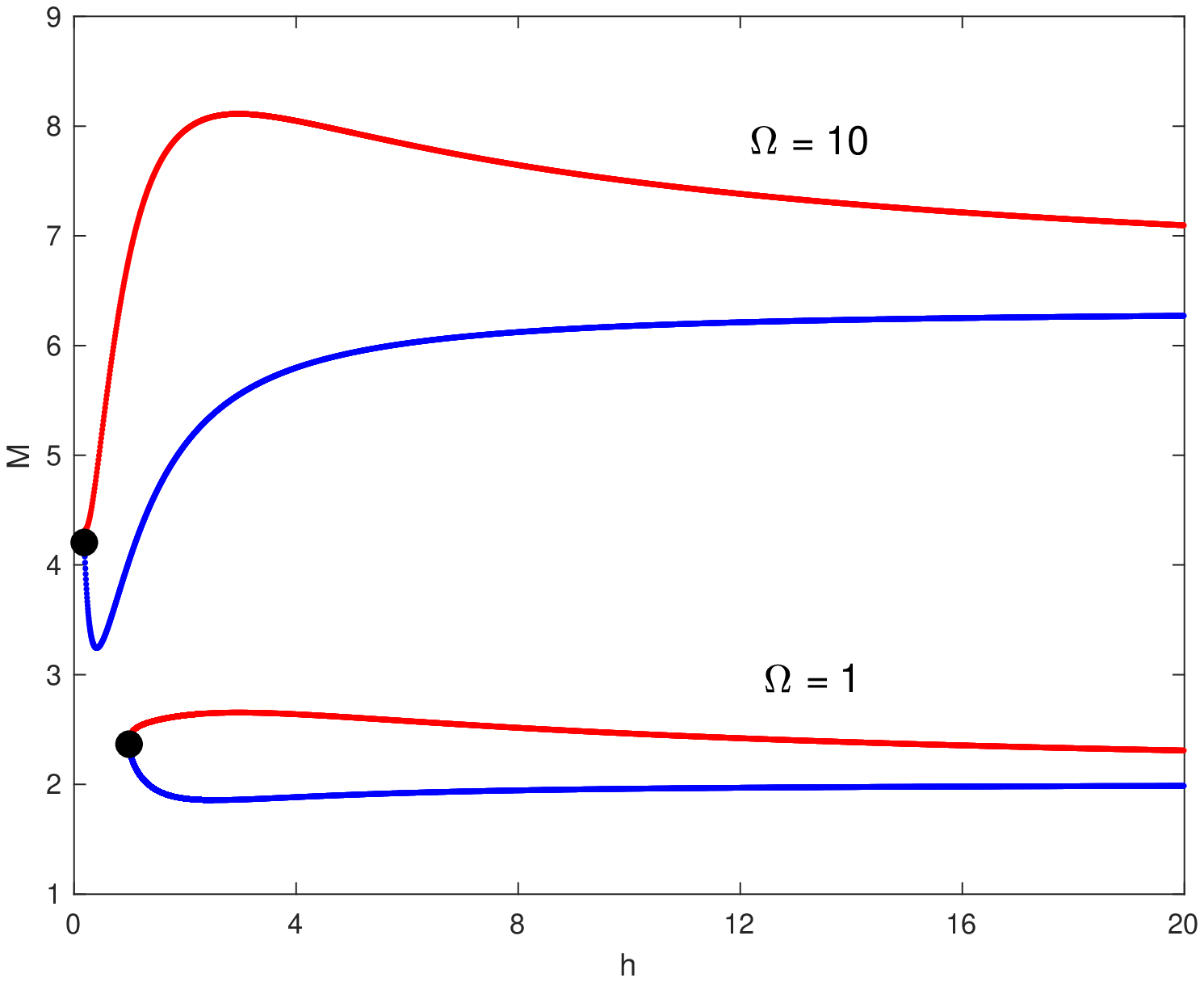}
	\caption{Branches of the single-site and double-site stationary states
for $p = 1$ (left) and $p = 3/2$ (right) with two values of $\Omega = 1$ and $\Omega = 10$.}
\label{fig-2}
\end{figure}

Figure \ref{fig-3} shows the same as Fig. \ref{fig-1}
for fixed $\Omega = 1$ with $p = 2$ (left) and $p = 3$ (right).
Compared to the case $p < 2$, the stationary states have two fold bifurcations both in the continuum limit $h \to 0$
and in the anti-continuum limit $h \to \infty$. This shows that the constraint $p < 2$
for persistence of stationary states in the anti-continuum limit $h \to \infty$
in Theorem \ref{thm-existence} is sharp.

\begin{figure}[h!]
	\centering
		\includegraphics[width=0.45\linewidth]{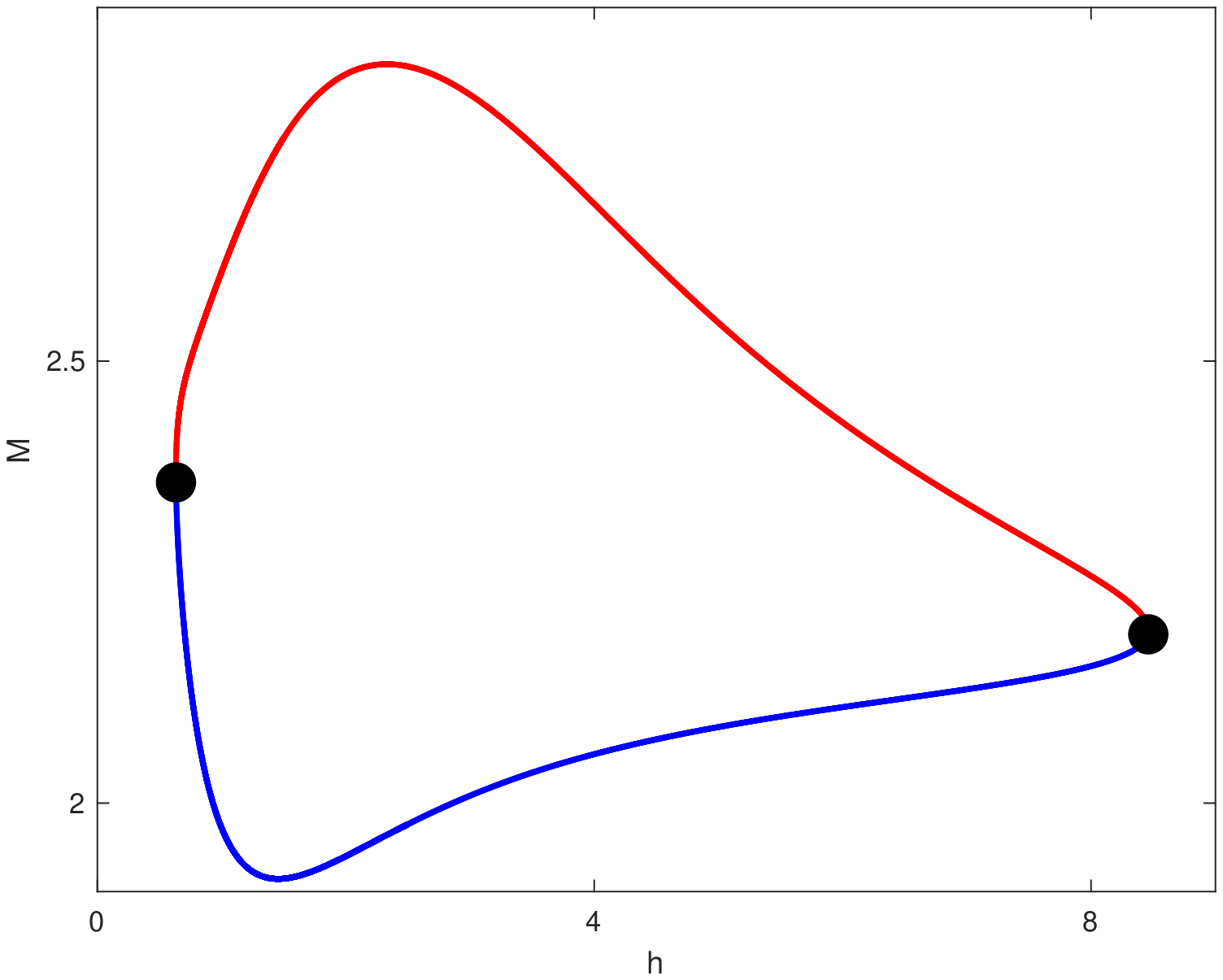}
		\includegraphics[width=0.45\linewidth]{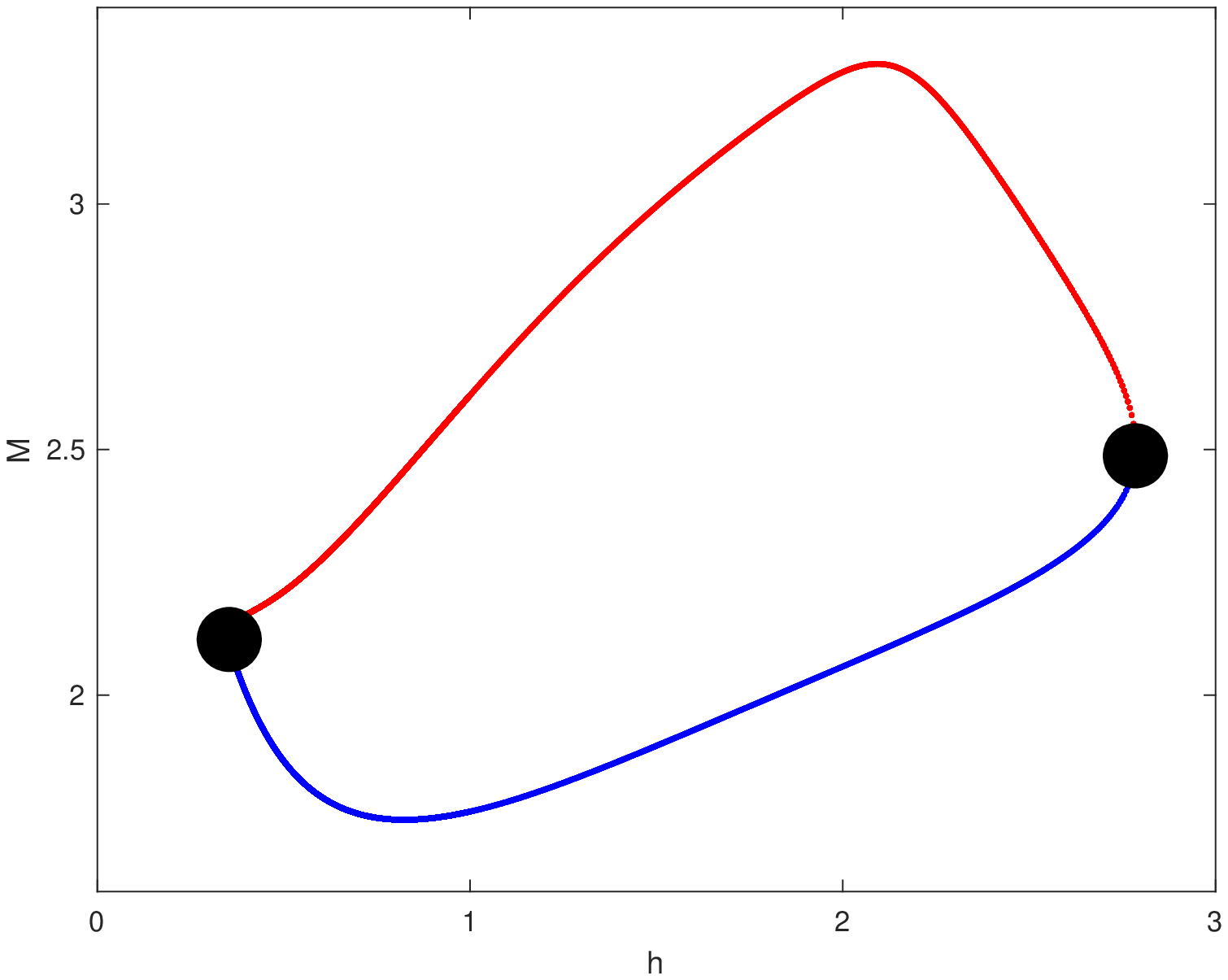}
\caption{The same as on Figure \ref{fig-1} but for $p = 2$ (left) or $p = 3$ (right) with $\Omega = 1$.}
\label{fig-3}
\end{figure}

Figure \ref{fig-4} shows two branches of the stationary states for $p = 2$ and two fixed values of $\Omega$
with $\Omega = 1$ and $\Omega = 2$. The critical value of $h$ for the fold bifurcation at smaller values of $h$ decreases in $\Omega$,
whereas that for the fold bifurcation at larger values of $h$ increases in $\Omega$.
This behavior can again be explained from the generalized scaling transformation \eqref{scaling}.

\begin{figure}[h!]
	\centering
		\includegraphics[width=0.8\linewidth]{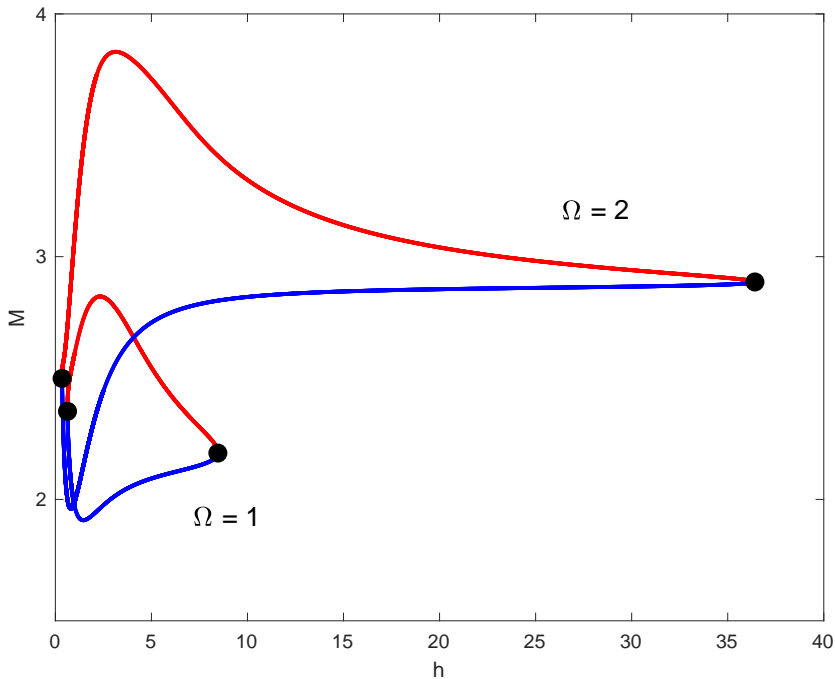}
\caption{Branches of the single-site and double-site stationary states
for $p = 2$ with two values of $\Omega = 1$ and $\Omega = 2$.}
\label{fig-4}
\end{figure}

In the exceptional case $p = 2$, we can fix $a = 1$ and rewrite the first equation of the system \eqref{eq-scaled} in the form:
\begin{equation}
\label{eq-final-2}
h^2 \left( \tilde{U}_{n+1}-2 \tilde{U}_{n}+ \tilde{U}_{n-1} \right) -
\tilde{\Omega} \tilde{U}_n + h^2 \tilde{G}_n \tilde{U}_n + \lambda |\tilde{U}_n|^{2p} \tilde{U}_n = 0,
\end{equation}
with $h^2 \to 0$ as $h \to 0$. If $\tilde{\Omega}_*(h)$
at the fold bifurcation in \eqref{eq-final-2} approaches
asymptotically to $\tilde{\Omega}_*(0) \neq 0$ as $h \to 0$, then $\Omega \sim h^{-4}$ so that
$\Omega \to \infty$ as $h \to 0$ and vice versa. Similarly, we can fix $a = -1$ and
rewrite the first equation of the system \eqref{eq-scaled} in the form:
\begin{equation}
\label{eq-final-3}
h^{-6} \left( \tilde{U}_{n+1}-2 \tilde{U}_{n}+ \tilde{U}_{n-1} \right) -
\tilde{\Omega} \tilde{U}_n + h^2 \tilde{G}_n \tilde{U}_n + \lambda |\tilde{U}_n|^{2p} \tilde{U}_n = 0,
\end{equation}
with $h^{-6} \to 0$ as $h \to \infty$. If $\tilde{\Omega}_*(h)$ at the fold bifurcation in \eqref{eq-final-3}
approaches asymptotically to $\tilde{\Omega}_*(\infty) < \infty$ as $h \to \infty$, then $\Omega \sim h^{4}$ so that
$\Omega \to \infty$ as $h \to \infty$ and vice versa. Thus, both dependencies of the critical value $h_*(\Omega)$ for
the two fold bifurcations seen on Figure \ref{fig-4} can be explained from the generalized scaling transformation
\eqref{scaling}.

Note that we do not give the numerical computations for two branches of stationary bound states for $p = 3$ and $\Omega = 2$.
Although we have found that the single-site stationary states have similar fold bifurcations in the direction of $h \to 0$ and
$h \to \infty$, we were able to connect them with the double-site stationary states at the left bifurcation point only.
At the right bifurcation point, both the single-site and double-site stationary states do not connect to each other but
connect to other stationary states of the system \eqref{eq-scaled}, which we were not able to detect numerically.

\section{Conclusion}
\label{sec-6}

We have proposed a gauge-invariant discrete CSS system \eqref{dcss} on a one-dimensional lattice. This system conserves the
mass \eqref{mass-discrete} but does not conserve the energy. By using the spatial gauge condition $A_1 \equiv 0$,
we have proven local and global well-posedness of the initial-value problem in $\ell^2$ for the scalar field $\phi$
and in $\ell^{\infty}$ for the gauge vector. We have also studied existence of the stationary bound states
from solutions of the coupled difference equations \eqref{equ23} in $\ell^1$ for the scalar field and
in $\ell^{\infty}$ for the gauge vector. For $p \in (0,2)$, we proved that the stationary bound states
persist in the anti-continuum limit $h \to \infty$ but do not persist in the continuum limit $h \to 0$. We have shown
numerically that the branch of single-site solutions terminates at a fold bifurcation with the branch of double-site solutions
as $h$ gets smaller. For $p \geq 2$, two fold bifurcations occur both as $h$ gets smaller and larger, so that the
stationary bound states do not persist both in the continuum limit $h \to 0$ and in the anti-continuum limit $h \to \infty$.

Among further problems, stability of stationary bound states is important for applications and interesting
mathematically. Due to the lack of the energy conservation, the methods of the Hamiltonian dynamical systems
for stability may not be applicable for the discrete CSS system \eqref{dcss} and new analytical methods need
to be developed for the stationary bound states of Theorem \ref{thm-existence}.

\vspace{0.25cm}

{\bf Acknowledgement:} The research of H. Huh was supported by LG Yonam Foundation of Korea and Basic Science Research Program
through the National Research Foundation of Korea funded by the Ministry of Education (2017R1D1A1B03028308).
The research of S. Hussain is supported by the NSERC USRA project. The research of D. Pelinovsky is supported
by the NSERC Discovery grant.

\vspace{0.25cm}

{\bf Conflict of interest.} The authors declare that they have no conflict of interest concerning
publication of this manuscript.

\end{document}